\pdfoutput=1
\RequirePackage{ifpdf}
\ifpdf % We are running pdfTeX in pdf mode
\documentclass[pdftex]{sigma}
\else
\documentclass{sigma}
\fi

\numberwithin{equation}{section}

\newtheorem{Theorem}{Theorem}[section]
\newtheorem{Corollary}[Theorem]{Corollary}
\newtheorem{Lemma}[Theorem]{Lemma}
\newtheorem{Proposition}[Theorem]{Proposition}

\newcommand{\al}{\alpha}
\newcommand{\be}{\beta}

\newcommand{\Ga}{\Gamma}

\newcommand{\bbZ}{\mathbb{Z}}
\newcommand{\bbC}{\mathbb{C}}

\begin{document}

\allowdisplaybreaks

\newcommand{\arXivNumber}{1601.05327}

\renewcommand{\PaperNumber}{051}

\FirstPageHeading

\ShortArticleName{Hypergeometric $\tau$ Functions of the $q$-Painlev\'e Systems of Types $A_4^{(1)}$ and $(A_1+A_1')^{(1)}$}

\ArticleName{Hypergeometric $\boldsymbol{\tau}$ Functions of the $\boldsymbol{q}$-Painlev\'e\\ Systems of Types $\boldsymbol{A_4^{(1)}}$ and $\boldsymbol{(A_1+A_1')^{(1)}}$}

\Author{Nobutaka NAKAZONO}

\AuthorNameForHeading{N.~Nakazono}

\Address{School of Mathematics and Statistics, The University of Sydney,\\ New South Wales 2006, Australia}
\Email{\href{mailto:nobua.n1222@gmail.com}{nobua.n1222@gmail.com}}
\URLaddress{\url{http://researchmap.jp/nakazono/}}

\ArticleDates{Received February 01, 2016, in f\/inal form May 16, 2016; Published online May 20, 2016}

\Abstract{We consider $q$-Painlev\'e equations arising from
birational representations of the extended af\/f\/ine Weyl groups of $A_4^{(1)}$- and $(A_1+A_1)^{(1)}$-types.
We study their hypergeometric solutions on the level of $\tau$ functions.}


\Classification{33D05; 33D15; 33E17; 39A13}

\section{Introduction}

\subsection{Purpose}

The purpose of this paper is to construct the hypergeometric $\tau$ functions associated with
$q$-Painlev\'e equations of $A_4^{(1)}$- and $A_6^{(1)}$-surface types in Sakai's classif\/ication~\cite{SakaiH2001:MR1882403}.
As a corollary, we obtain the hypergeometric solutions of the corresponding $q$-Painlev\'e equations.

This work is motivated by the project to construct all possible hypergeometric $\tau$ functions associated with
the multiplicative surface types in the Sakai's classif\/ication~\cite{SakaiH2001:MR1882403}, that is,
$A_0^{(1)}$-, $A_1^{(1)}$-, $A_2^{(1)}$-, $A_3^{(1)}$-, $A_4^{(1)}$-, $A_5^{(1)}$- and $A_6^{(1)}$-surface types.
The corresponding symmetry groups are $W\big(E_8^{(1)}\big)$, $\widetilde{W}\big(E_7^{(1)}\big)$, $\widetilde{W}\big(E_6^{(1)}\big)$, $\widetilde{W}\big(D_5^{(1)}\big)$, $\widetilde{W}\big(A_4^{(1)}\big)$, $\widetilde{W}\big((A_2+A_1)^{(1)}\big)$ and $\widetilde{W}\big((A_1+A_1')^{(1)}\big)$, respectively.
The works for $W\big(E_8^{(1)}\big)$-type \cite{MasudaT2011:MR2765599}, $\widetilde{W}\big(E_7^{(1)}\big)$-type~\cite{MasudaT2009:MR2506177} and $\widetilde{W}\big((A_2+A_1)^{(1)}\big)$-type~\cite{NakazonoN2010:MR2769931} have been done.
In this paper, we consider the hypergeometric $\tau$ functions of $\widetilde{W}\big(A_4^{(1)}\big)$- and $\widetilde{W}\big((A_1+A_1')^{(1)}\big)$-types.

\subsection{Background}

Discrete Painlev\'e equations are nonlinear ordinary dif\/ference equations of second order,
which include discrete analogues of the six Painlev\'e equations,
and are classif\/ied by types of rational surfaces connected to af\/f\/ine Weyl groups~\cite{SakaiH2001:MR1882403}.
They admit particular solutions, so called hypergeometric solutions,
which are expressible in terms of the hypergeometric type functions, when some of the parameters take special values
(see, for example, \cite{KMNOY2004:MR2077840,KMNOY2005:MR2153786,KN2015:MR3340349} and references therein).
Together with the Painlev\'e equations, discrete Painlev\'e equations are now regarded as one of the most important classes of equations
in the theory of integrable systems (see, e.g., \cite{GR2004:MR2087743,KNY2015:arXiv150908186K}).

It is well known that the $\tau$ functions play a crucial role in the theory of integrable systems~\cite{book_MJD2000:MR1736222}, and
it is also possible to introduce them in the theory of Painlev\'e systems \cite{JM1981:MR625446,JM1981:MR636469,JMU1981:MR630674,book_NoumiM2004:MR2044201,OkamotoK1987:MR916698,OkamotoK1987:MR914314,OkamotoK1986:MR854008,OkamotoK1987:MR927186}.
A representation of the af\/f\/ine Weyl groups can be lifted on the level of the $\tau$ functions
\cite{JNS2015:MR3403054,JNS:paper5,KMNOY2003:MR1984002,KMNOY2006:MR2353465,MasudaT2009:MR2506177,MasudaT2011:MR2765599,TsudaT2006:MR2207047,TM2006:MR2202304},
which gives rise to various bilinear equations of Hirota type satisf\/ied by the $\tau$ functions.

Usually, the hypergeometric solutions of discrete Painlev\'e equations are derived by reducing the bilinear equations to
the Pl\"ucker relations by using the contiguity relations satisf\/ied by the entries of determinants
\cite{HK2007:MR2392886,HKW2006:MR2264726,JKM2006:MR2297948,KK2003:MR1952868,KM1999:MR1694666,KNY2001:MR1876614,KO1998:MR1629434,KOS1995:MR1341981,KOSGR1994:MR1267458,ON1992:MR1200649,SakaiH1998:MR1632613}.
This method is elementary, but it encounters technical dif\/f\/iculties for discrete Painlev\'e equations with large symmetries.
In order to overcome this dif\/f\/iculty, Masuda has proposed a method of constructing hypergeometric solutions
under a certain boundary condition on the lattice where the $\tau$ functions live, so that they are consistent with the action of the af\/f\/ine Weyl groups. We call such hypergeometric solutions hypergeometric $\tau$ functions \cite{MasudaT2009:MR2506177,MasudaT2011:MR2765599,NakazonoN2010:MR2769931}.
Although this requires somewhat complex calculations, the merit is that it is systematic and can be applied to the systems with large symmetries.

Some discrete Painlev\'e equations have been found in the studies of random matrices \cite{BK1990:MR1040213,IZ1980:MR562985,PS1996:Unitary-matrix}.
As one such example, let us consider the partition function of the Gaussian Unitary Ensemble of an $n\times n$ random matrix:
\begin{gather*}
 Z_n^{(2)}=\int^\infty_{-\infty}\cdots\int^\infty_{-\infty}\Delta(t_1,\dots,t_n)^2\prod_{i=1}^n{\rm e}^{-g_1{t_i}^2-g_2{t_i}^4}{\rm d}t_i,
\end{gather*}
where $g_2>0$ and $\Delta(t_1,\dots,t_n)$ is Vandermonde's determinant. Letting
\begin{gather*}
 R_n=\frac{Z_{n+1}^{(2)}Z_{n-1}^{(2)}}{\big(Z_n^{(2)}\big)^2},
\end{gather*}
we obtain the following dif\/ference equation \cite{BK1990:MR1040213,GORS1998:MR1626199,GRP1991:MR1125950,RGH1991:MR1125951}
\begin{gather}\label{eqn:intro_dp1}
 R_{n+1}+R_n+R_{n-1}=\frac{n}{4g_2} \frac{1}{R_n}-\frac{g_1}{2g_2}.
\end{gather}
Equation \eqref{eqn:intro_dp1} is known as a discrete analogue of the Painlev\'e I equation
and also as a B\"acklund transformation of the Painlev\'e IV equation.
The partition function~$Z_n^{(2)}$ corresponds to hypergeometric $\tau$ functions.
Such relations between discrete Painlev\'e equations and random matrices are well investigated.
Moreover, in recent years, the relations between $\tau$ functions of Painlev\'e systems and
a certain class of integrable partial dif\/ference equations introduced by Adler--Bobenko--Suris and Boll
\cite{ABS2003:MR1962121,ABS2009:MR2503862,BollR2011:MR2846098,BollR2012:MR3010833,BollR:thesis},
which include a discrete analogue of the Korteweg--de Vries equation,
are well investigated \cite{AHKM2014:MR3221837,HKM2011:MR2788707,JNS:paper3,JNS2015:MR3403054,JNS:paper5}.
Throughout these relations and by using the hypergeometric $\tau$ functions,
a discrete analogue of the power function was derived and its properties,
such as discrete analogue of the Riemann surface and circle packing, were shown in
\cite{AgafonovSI2003:MR2021726,AgafonovSI2003:MR1957234,AB2000:MR1747617,AB2003:MR2006759,AHKM2014:MR3221837,NRGO2001:MR1819383}.
These results consolidate the importance of the studies of the hypergeometric $\tau$ function for applications of Painlev\'e systems.

In \cite{HK2007:MR2392886,HKW2006:MR2264726}, the hypergeometric solutions of the $q$-Painlev\'e equations \eqref{eqn:qp5_1} and~\eqref{eqn:A1A1_qp2_1} (or~\eqref{eqn:A1A1_qp2_2}) are constructed by solving the minimum required bilinear equations to obtain those equations.
In this paper, we solve all bilinear equations arising from the actions of the translation subgroups of $\widetilde{W}\big(A_4^{(1)}\big)$ and $\widetilde{W}\big((A_1+A_1')^{(1)}\big)$, that is, the hypergeometric $\tau$ functions given in Theo\-rems~\ref{theorem:A4_hypertau} and~\ref{theorem:A1A1_hypertau} are for not only the hypergeometric solutions of the $q$-Painlev\'e equations~\eqref{eqn:qp5_1} and~\eqref{eqn:A1A1_qp2_1} but also those of other $q$-Painlev\'e equations, e.g.,~\eqref{eqn:qp5_2},~\eqref{eqn:A1A1_qP_T2} and~\eqref{eqn:A1A1_qP_T3} (see Corollaries~\ref{corollary:A4_HGsol} and~\ref{corollary:A1A1_HGsol}).
Moreover, as mentioned above we can derive the various integrable partial dif\/ference equations from the
$\tau$ functions of discrete Painlev\'e equations (see, for example,~\cite{HKM2011:MR2788707,JNS2015:MR3403054,JNS:paper5}).
Therefore, the hypergeometric $\tau$ functions constructed in this paper also give the hypergeometric solutions of the partial dif\/ference equations appeared in~\cite{JNS2015:MR3403054, JNS:paper5}.

\subsection{Plan of the paper}

This paper is organized as follows: in Section~\ref{section:HGtau_A4}, we f\/irst introduce $\tau$ functions with the representation of the af\/f\/ine Weyl group $\widetilde{W}\big(A_4^{(1)}\big)$. Next, we construct the hypergeometric $\tau$ functions of $\widetilde{W}\big(A_4^{(1)}\big)$-type (see Theorem~\ref{theorem:A4_hypertau}). Finally, we obtain the hypergeometric solutions of the $q$-Painlev\'e equations of $A_4^{(1)}$-surface type (see Corollary~\ref{corollary:A4_HGsol}). In Section~\ref{section:HGtau_A1A1}, we summarize the result for the $\widetilde{W}\big((A_1+A_1')^{(1)}\big)$-type (or, $A_6^{(1)}$-surface type).

\subsection[$q$-Special functions]{$\boldsymbol{q}$-Special functions}
We use the following conventions of $q$-analysis with $|p|,|q|<1$ throughout this paper \cite{book_GR2004:MR2128719}.
\begin{itemize}\itemsep=0pt
\item $q$-Shifted factorials:
\begin{gather*}
(a;q)_n=\prod_{i=0}^{n-1}\big(1-q^ia\big),\qquad n=1,2,\dots,\qquad (a;q)_\infty=\prod_{i=0}^\infty\big(1-q^ia\big),\\
(a;p,q)_\infty=\prod_{i,j=0}^\infty\big(1-q^ip^ja\big).
\end{gather*}
\item Modif\/ied Jacobi theta function:
\begin{gather*}
 \Theta(a;q)=(a;q)_\infty\big(qa^{-1};q\big)_\infty.
\end{gather*}
\item Elliptic gamma function:
\begin{gather*}
 \Ga(a;p,q)=\frac{\big(pqa^{-1};p,q\big)_{\infty}}{(a;p,q)_{\infty}}.
\end{gather*}
\item Basic hypergeometric series:
\begin{gather*}
 {}_s\varphi_r\left(\begin{matrix}a_1,\dots,a_s\\b_1,\dots,b_r\end{matrix};q,z\right)
 =\sum_{n=0}^{\infty}\frac{(a_1,\dots,a_s;q)_n}{(b_1,\dots,b_r;q)_n(q;q)_n} \big[(-1)^nq^{n(n-1)/2}\big] ^{1+r-s}z^n,
\end{gather*}
where
\begin{gather*}
 (a_1,\dots,a_s;q)_n=\prod_{i=1}^s(a_i;q)_n.
\end{gather*}
\end{itemize}
We note that the following formulae hold
\begin{alignat*}{3}
 &\frac{(q^na;q)_\infty}{(a;q)_\infty}=\prod_{i=0}^{n-1}\frac{1}{1-q^ia},\qquad
 &&\frac{\Theta(q^na;q)}{\Theta(a;q)}=(-1)^n\prod_{i=0}^{n-1}\frac{1}{q^ia},&\\
 &\frac{(q^na;p,q)_\infty}{(a;p,q)_\infty}=\prod_{i=0}^{n-1}\frac{1}{(q^ia;p)_\infty},\qquad
 &&\frac{(p^na;p,q)_\infty}{(a;p,q)_\infty}=\prod_{i=0}^{n-1}\frac{1}{(p^ia;q)_\infty},&\\
 &\frac{\Ga(q^na;p,q)}{\Ga(a;p,q)}=\prod_{i=0}^{n-1}\Theta\big(q^ia;p\big),\qquad
 &&\frac{\Ga(p^na;p,q)}{\Ga(a;p,q)}=\prod_{i=0}^{n-1}\Theta\big(p^ia;q\big),&
\end{alignat*}
where $n\in\bbZ_{>0}$.

\section[Hypergeometric $\tau$ functions of $\widetilde{W}\big(A_4^{(1)}\big)$-type]{Hypergeometric $\boldsymbol{\tau}$ functions of $\boldsymbol{\widetilde{W}\big(A_4^{(1)}\big)}$-type}\label{section:HGtau_A4}

In this section, we construct the hypergeometric $\tau$ functions of $\widetilde{W}\big(A_4^{(1)}\big)$-type.

\subsection[$\tau$ functions]{$\boldsymbol{\tau}$ functions}
Let us consider ten variables: $\tau_i^{(j)}$ $(i=1,2$, $j=1,\dots,5)$ and
six parameters: $a_0,\dots,a_4,q\in\bbC^\ast$ with the following three relations for the variables
\begin{subequations}\label{eqns:A4_conditions_tau}
\begin{gather}
 \tau_2^{(1)}=\frac{a_0 a_1 \big(a_3 \tau_1^{(3)} \tau_1^{(5)}+a_0 \tau_1^{(4)} \tau_2^{(3)}\big)}{a_2 {a_3}^2 \tau_2^{(5)}},
 \label{eqns:A4_conditions_tau_1}\\
 \tau_2^{(2)}=\frac{a_1 a_2 \big(a_4 \tau_1^{(1)} \tau_1^{(4)}+a_1 \tau_1^{(5)} \tau_2^{(4)}\big)}{a_3 {a_4}^2 \tau_2^{(1)}},
 \label{eqns:A4_conditions_tau_2}\\
 \tau_2^{(4)}=\frac{a_3 a_4 \big(a_1 \tau_1^{(1)} \tau_1^{(3)}+a_3 \tau_1^{(2)} \tau_2^{(1)}\big)}{a_0 {a_1}^2 \tau_2^{(3)}},
 \label{eqns:A4_conditions_tau_3}
\end{gather}
\end{subequations}
and the following condition for the parameters
\begin{gather*}
 a_0a_1a_2a_3a_4=q.
\end{gather*}
The action of the transformation group $\langle s_0,s_1,s_2,s_3,s_4,\sigma,\iota\rangle$ on the parameters is given by
\begin{gather*}
 s_i(a_j)=a_j{a_i}^{-a_{ij}},\qquad \sigma(a_i)=a_{i+1},\\
 \iota\colon \ (a_0,a_1,a_2,a_3,a_4)\mapsto \big({a_0}^{-1},{a_4}^{-1},{a_3}^{-1},{a_2}^{-1},{a_1}^{-1}\big),
\end{gather*}
where $i,j\in\mathbb{Z}/5\mathbb{Z}$ and the symmetric $5\times 5$ matrix
\begin{gather*}
 (a_{ij})_{i,j=0}^4
 =\left(\begin{matrix}
 2&-1&0&0&-1\\-1&2&-1&0&0\\0&-1&2&-1&0\\0&0&-1&2&-1\\-1&0&0&-1&2\end{matrix}\right)
\end{gather*}
is the Cartan matrix of type $A_4^{(1)}$. Moreover, the action on the variables is given by
\begin{subequations}\label{eqns:A4_weylaction_tau}
\begin{gather}
 s_i\big(\tau_1^{(i+5)}\big)=\tau_2^{(i+4)},\qquad s_i\big(\tau_2^{(i+3)}\big)
 =\frac{a_{i+3} a_{i+4}\big(a_i a_{i+1} \tau_1^{(i+1)}\! \tau_1^{(i+3)}+a_{i+3} \tau_1^{(i+2)}\! \tau_2^{(i+1)}\big)}
 {{a_{i+1}}^2 \tau_1^{(i+5)}},\!\!\!\!\\
s_i\big(\tau_2^{(i+4)}\big)=\tau_1^{(i+5)},\qquad s_i\big(\tau_2^{(i+5)}\big)
 =\frac{a_{i+4}\big(a_{i+2} \tau_1^{(i+2)} \tau_1^{(i+4)}+a_i a_{i+4} \tau_1^{(i+3)} \tau_2^{(i+2)}\big)}
 {a_i a_{i+1} {a_{i+2}}^2 \tau_1^{(i+5)}},\\
 \sigma\big(\tau_1^{(i)}\big)=\tau_1^{(i+1)},\qquad \sigma\big(\tau_2^{(i)}\big)=\tau_2^{(i+1)},\\
\iota\colon \ \big(\tau_1^{(1)}\!,\tau_1^{(2)}\!,\tau_1^{(3)}\!,\tau_1^{(4)}\!,\tau_2^{(1)}\!,\tau_2^{(2)}\!,\tau_2^{(3)}\!,\tau_2^{(5)}\big)
 \mapsto\big(\tau_1^{(4)}\!,\tau_1^{(3)}\!,\tau_1^{(2)}\!,\tau_1^{(1)}\!,\tau_2^{(2)}\!,\tau_2^{(1)}\!,\tau_2^{(5)}\!,\tau_2^{(3)}\big),
\end{gather}
\end{subequations}
where $i\in\mathbb{Z}/5\mathbb{Z}$. In general, for a function $F=F\big(a_i,\tau_j^{(k)}\big)$, we let an element
$w\in\widetilde{W}\big(A_4^{(1)}\big)$ act as $w.F=F\big(w.a_i,w.\tau_j^{(k)}\big)$, that is, $w$ acts on the arguments from the left.
Note that $q=a_0a_1a_2a_3a_4$ is invariant under the action of $\langle s_0,s_1,s_2,s_3,s_4,\sigma\rangle$.

\begin{Proposition}[\cite{JNS:paper5,TsudaT2006:MR2207047}]
The group of birational transformations $\langle s_0,s_1,s_2,s_3,s_4,\sigma,\iota\rangle$,
denoted by $\widetilde{W}\big(A_4^{(1)}\big)$, forms the extended affine Weyl group of type $A_4^{(1)}$.
Namely, the transformations satisfy the fundamental relations
\begin{gather*}
 {s_i}^2=1,\qquad (s_is_{i\pm 1})^3=1,\qquad (s_is_j)^2=1,\qquad j\neq i\pm 1,\\
 \sigma^5=1,\qquad \sigma s_i=s_{i+1}\sigma,\qquad \iota^2=1,\qquad \iota s_0=s_0\iota,\qquad
 \iota s_1=s_4\iota,\qquad \iota s_2=s_3\iota,
\end{gather*}
where $i,j\in\mathbb{Z}/5\mathbb{Z}$.
\end{Proposition}

To iterate each variable $\tau_i^{(j)}$, we need the translations $T_i$, $i=0,\dots,4$, def\/ined by
\begin{subequations}\label{eqn:A4_translations}
\begin{gather}
 T_0=\sigma s_4s_3s_2s_1,\qquad T_1=\sigma s_0s_4s_3s_2,\qquad T_2=\sigma s_1s_0s_4s_3,\qquad T_3=\sigma s_2s_1s_0s_4,\\
 T_4=\sigma s_3s_2s_1s_0.
\end{gather}
\end{subequations}
The action of translations on the parameters is given by
\begin{gather*}
 T_i(a_i)=qa_i,\qquad T_i(a_{i+1})=q^{-1}a_{i+1},
\end{gather*}
where $i\in\bbZ/5\bbZ$.
Note that $T_i$, $i=0,\dots,4$, commute with each other and
\begin{gather*}
 T_0T_1T_2T_3T_4=1.
\end{gather*}
We def\/ine $\tau$ functions by
\begin{gather}\label{eqn:A4_tau_l_0123}
 \tau_{l_1}^{l_0,l_2,l_3}={T_0}^{l_0}{T_1}^{l_1}{T_2}^{l_2}{T_3}^{l_3}\big(\tau_2^{(3)}\big),
\end{gather}
where $l_i\in\bbZ$.
We note that
\begin{subequations}\label{eqns:A4_config_tau}
\begin{gather}
 \tau_1^{(1)}=\tau_0^{1,0,1},\qquad
 \tau_1^{(2)}=\tau_1^{1,0,1},\qquad
 \tau_1^{(3)}=\tau_1^{1,1,1},\qquad
 \tau_1^{(4)}=\tau_1^{1,1,2},\qquad
 \tau_1^{(5)}=\tau_0^{0,0,1},\\
 \tau_2^{(1)}=\tau_0^{1,1,1},\qquad
 \tau_2^{(2)}=\tau_1^{1,0,2},\qquad
 \tau_2^{(3)}=\tau_0^{0,0,0},\qquad
 \tau_2^{(4)}=\tau_1^{2,1,2},\qquad
 \tau_2^{(5)}=\tau_1^{0,0,1}.
\end{gather}
\end{subequations}

\subsection [Hypergeometric $\tau$ functions]{Hypergeometric $\boldsymbol{\tau}$ functions}\label{subsection:HGtau_A4}

The aim of this section is to construct the hypergeometric $\tau$ functions of $\widetilde{W}\big(A_4^{(1)}\big)$-type.

Hereinafter, we consider the $\tau$ functions $\tau_{l_1}^{l_0,l_2,l_3}$ satisfying the following conditions:
\begin{enumerate}\itemsep=0pt
\item[(i)]
$\tau_{l_1}^{l_0,l_2,l_3}$ satisfy the action of the translation subgroup of $\widetilde{W}\big(A_4^{(1)}\big)$, $\langle T_0,T_1,T_2,T_3,T_4\rangle$;
\item[(ii)]
$\tau_{l_1}^{l_0,l_2,l_3}$ are functions in $a_0$, $a_2$ and $a_4$ consistent with
the action of $\langle T_0,T_2,T_3\rangle$, i.e.,
$ \tau_{l_1}^{l_0,l_2,l_3}=\tau_{l_1}\big(q^{l_0}a_0,q^{l_2}a_2,q^{-l_3}a_4\big)$;
\item[(iii)]
$\tau_{l_1}^{l_0,l_2,l_3}$ satisfy the following boundary conditions:
\begin{gather}\label{eqn:A4_boundary_tau}
 \tau_{l_1}^{l_0,l_2,l_3}=0,
\end{gather}
for $l_1<0$;
\end{enumerate}
under the conditions of parameters
\begin{gather}\label{eqn:condition_para}
 a_0a_1=q.
\end{gather}
We here call such functions $\tau_{l_1}^{l_0,l_2,l_3}$ hypergeometric $\tau$ functions of $\widetilde{W}\big(A_4^{(1)}\big)$-type.

From the condition (i), every $\tau_{l_1}^{l_0,l_2,l_3}$ can be given by
a rational function of ten variab\-les~$\tau_i^{(j)}$ (or, $\big\{\tau_0^{l_0,l_2,l_3}\big\}_{l_i\in\bbZ}$ and $\big\{\tau_1^{l_0,l_2,l_3}\big\}_{l_i\in\bbZ}$).
Therefore, our purpose in this section is to obtain the explicit formulae for $\big\{\tau_0^{l_0,l_2,l_3}\big\}_{l_i\in\bbZ}$ and $\big\{\tau_1^{l_0,l_2,l_3}\big\}_{l_i\in\bbZ}$,
satisfying the condition~(ii) under the condition~(iii) and construct the closed-form expressions of $\big\{\tau_{l_1}^{l_0,l_2,l_3}\big\}_{l_i\in\bbZ,\,l_1\geq2}$.

{\bf Step 1.}
Begin by preparing the equations necessary for the construction of the hypergeometric $\tau$ functions of $\widetilde{W}\big(A_4^{(1)}\big)$-type.
From the actions~\eqref{eqns:A4_weylaction_tau} and the def\/initions~\eqref{eqn:A4_translations},
the actions of~ $T_0$,~$T_2$ and~$T_3$ and their inverses on ten variables $\tau_i^{(j)}$ are given by the following
\begin{subequations}\label{eqns:action_T0_tau}
\begin{gather}
 T_0\big(\tau_1^{(4)}\big)=\tau_2^{(4)},\qquad T_0\big(\tau_1^{(5)}\big)=\tau_1^{(1)},\qquad T_0\big(\tau_2^{(5)}\big)=\tau_1^{(2)},\nonumber\\
 T_2\big(\tau_1^{(1)}\big)=\tau_2^{(1)},\qquad T_2\big(\tau_1^{(2)}\big)=\tau_1^{(3)},\qquad T_2\big(\tau_2^{(2)}\big)=\tau_1^{(4)},\nonumber\\
 T_3\big(\tau_1^{(2)}\big)=\tau_2^{(2)},\qquad T_3\big(\tau_1^{(3)}\big)=\tau_1^{(4)},\qquad T_3\big(\tau_2^{(3)}\big)=\tau_1^{(5)},\nonumber\\
T_0\big(\tau_1^{(1)}\big)
 =\frac{q{a_0}^2 a_4 \big(a_3 \tau_1^{(1)} T_0\big(\tau_1^{(3)}\big)+a_0 a_1 \tau_2^{(4)} T_0\big(\tau_2^{(3)}\big)\big)}
 {a_3\tau_1^{(3)}},\label{eqn:action_T0_tau_1}\\
T_0\big(\tau_1^{(2)}\big)
 =\frac{a_0 a_1 \big(qa_0 \tau_2^{(4)} T_0\big(\tau_2^{(3)}\big)+a_2 a_3 \tau_1^{(1)} T_0\big(\tau_1^{(3)}\big)\big)}
 {{a_3}^2 \tau_2^{(1)}},\label{eqn:action_T0_tau_2}\\
T_0\big(\tau_1^{(3)}\big)
 =\frac{a_3 a_4 \big(a_0 a_1 \tau_1^{(1)} \tau_1^{(3)}+a_3 \tau_1^{(2)} \tau_2^{(1)}\big)}
 {{a_1}^2 \tau_1^{(5)}},\label{eqn:action_T0_tau_3}\\
T_0\big(\tau_2^{(1)}\big)
 =\frac{a_0 a_1 \big(qa_0 \tau_2^{(4)} T_0\big(\tau_2^{(3)}\big)+a_3 \tau_1^{(1)} T_0\big(\tau_1^{(3)}\big)\big)}
 {a_2 {a_3}^2 \tau_1^{(2)}},\label{eqn:action_T0_tau_4}\\
T_0\big(\tau_2^{(2)}\big)
 =\frac{a_1 a_2 \big(q^{-1}a_1 \tau_1^{(1)} T_0\big(\tau_2^{(4)}\big)+a_4 \tau_2^{(4)} T_0\big(\tau_1^{(1)}\big)\big)}
 {qa_3 {a_4}^2 T_0\big(\tau_2^{(1)}\big)},\label{eqn:action_T0_tau_5}\\
T_0\big(\tau_2^{(3)}\big)
 =\frac{a_3 \big(a_1 \tau_1^{(1)} \tau_1^{(3)}+a_3 a_4 \tau_1^{(2)} \tau_2^{(1)}\big)}
 {a_0 {a_1}^2 a_4 \tau_1^{(4)}},\label{eqn:action_T0_tau_6}\\
T_0\big(\tau_2^{(4)}\big)
 =\frac{a_3 a_4 \big(a_1 T_0\big(\tau_1^{(1)}\big) T_0\big(\tau_1^{(3)}\big)+qa_3 T_0\big(\tau_1^{(2)}\big) T_0\big(\tau_2^{(1)}\big)\big)}
 {a_0 {a_1}^2 T_0\big(\tau_2^{(3)}\big)},\label{eqn:action_T0_tau_7}\\
{T_0}^{-1}\big(\tau_1^{(3)}\big)
 =\frac{a_0 \big(a_3 \tau_1^{(3)} \tau_1^{(5)}+a_0 a_1 \tau_1^{(4)} \tau_2^{(3)}\big)}
 {a_1 a_2 {a_3}^2 \tau_1^{(1)}},\label{eqn:action_T0_tau_8}\\
{T_0}^{-1}\big(\tau_1^{(4)}\big)
 =\frac{a_3 \big(qa_1 \tau_1^{(5)} {T_0}^{-1}\big(\tau_1^{(3)}\big)+a_3 a_4 \tau_2^{(5)} {T_0}^{-1}\big(\tau_2^{(1)}\big)\big)}
 {qa_0 {a_1}^2 a_4 \tau_2^{(3)}},\label{eqn:action_T0_tau_9}\\
{T_0}^{-1}\big(\tau_1^{(5)}\big)
 =\frac{a_3 a_4 \big(a_0 a_1 \tau_1^{(5)} {T_0}^{-1}\big(\tau_1^{(3)}\big)+a_3 \tau_2^{(5)} {T_0}^{-1}\big(\tau_2^{(1)}\big)\big)}
 {q^2{a_1}^2 \tau_1^{(3)}},\label{eqn:action_T0_tau_10}\\
{T_0}^{-1}\big(\tau_2^{(1)}\big)
 =\frac{a_0 a_1 \big(a_2 a_3 \tau_1^{(3)} \tau_1^{(5)}+a_0 \tau_1^{(4)} \tau_2^{(3)}\big)}{{a_3}^2 \tau_1^{(2)}},\label{eqn:action_T0_tau_11}\\
{T_0}^{-1}\big(\tau_2^{(2)}\big)
 =\frac{qa_1 a_2 \big(qa_1 \tau_1^{(4)} {T_0}^{-1}\big(\tau_1^{(5)}\big)+a_4 \tau_1^{(5)} {T_0}^{-1}\big(\tau_1^{(4)}\big)\big)}
 {a_3 {a_4}^2 {T_0}^{-1}\big(\tau_2^{(1)}\big)},\label{eqn:action_T0_tau_12}\\
{T_0}^{-1}\big(\tau_2^{(3)}\big)
 =\frac{a_3 a_4 \big(qa_1 \tau_1^{(5)} {T_0}^{-1}\big(\tau_1^{(3)}\big)+a_3 \tau_2^{(5)} {T_0}^{-1}\big(\tau_2^{(1)}\big)\big)}
 {qa_0 {a_1}^2 \tau_1^{(4)}},\label{eqn:action_T0_tau_13}\\
{T_0}^{-1}\big(\tau_2^{(5)}\big)
 =\frac{a_0 a_1 \big(q^{-1}a_0 {T_0}^{-1}\big(\tau_1^{(4)}\big) {T_0}^{-1}\big(\tau_2^{(3)}\big)+a_3 {T_0}^{-1}\big(\tau_1^{(3)}\big) {T_0}^{-1}\big(\tau_1^{(5)}\big)\big)}
 {a_2 {a_3}^2 {T_0}^{-1}\big(\tau_2^{(1)}\big)},\label{eqn:action_T0_tau_14}
\end{gather}
\end{subequations}
\vspace{-9mm}
\begin{subequations}\label{eqns:action_T2_tau}
\begin{gather}
T_2\big(\tau_1^{(3)}\big)
 =\frac{q{a_2}^2 a_1 \big(a_0 \tau_1^{(3)} T_2\big(\tau_1^{(5)}\big)+a_2 a_3 \tau_2^{(1)} T_2\big(\tau_2^{(5)}\big)\big)}
 {a_0\tau_1^{(5)}},\label{eqn:action_T2_tau_1}\\
T_2\big(\tau_1^{(4)}\big)
 =\frac{a_2 a_3 \big(qa_2\tau_2^{(1)} T_2\big(\tau_2^{(5)}\big)+a_4 a_0 \tau_1^{(3)} T_2\big(\tau_1^{(5)}\big)\big)}
 {{a_0}^2 \tau_2^{(3)}},\label{eqn:action_T2_tau_2}\\
T_2\big(\tau_1^{(5)}\big)
 =\frac{a_0 a_1 \big(a_2 a_3 \tau_1^{(3)} \tau_1^{(5)}+a_0 \tau_1^{(4)} \tau_2^{(3)}\big)}
 {{a_3}^2 \tau_1^{(2)}},\label{eqn:action_T2_tau_3}\\
T_2\big(\tau_2^{(3)}\big)
 =\frac{a_2 a_3 \big(qa_2 \tau_2^{(1)} T_2\big(\tau_2^{(5)}\big)+a_0 \tau_1^{(3)} T_2\big(\tau_1^{(5)}\big)\big)}
 {a_4 {a_0}^2 \tau_1^{(4)}},\label{eqn:action_T2_tau_4}\\
T_2\big(\tau_2^{(4)}\big)
 =\frac{a_3 a_4 \big(q^{-1}a_3 \tau_1^{(3)} T_2\big(\tau_2^{(1)}\big)+a_1 \tau_2^{(1)} T_2\big(\tau_1^{(3)}\big)\big)}
 {qa_0 {a_1}^2 T_2\big(\tau_2^{(3)}\big)},\label{eqn:action_T2_tau_5}\\
T_2\big(\tau_2^{(5)}\big)
 =\frac{a_0 \big(a_3 \tau_1^{(3)} \tau_1^{(5)}+a_0 a_1 \tau_1^{(4)} \tau_2^{(3)}\big)}
 {a_2 {a_3}^2 a_1 \tau_1^{(1)}},\label{eqn:action_T2_tau_6}\\
T_2\big(\tau_2^{(1)}\big)
 =\frac{a_0 a_1 \big(a_3 T_2\big(\tau_1^{(3)}\big) T_2\big(\tau_1^{(5)}\big)+qa_0 T_2\big(\tau_1^{(4)}\big) T_2\big(\tau_2^{(3)}\big)\big)}
 {a_2 {a_3}^2 T_2\big(\tau_2^{(5)}\big)},\label{eqn:action_T2_tau_7}\\
{T_2}^{-1}\big(\tau_1^{(5)}\big)
 =\frac{a_2 \big(a_0 \tau_1^{(5)} \tau_1^{(2)}+a_2 a_3 \tau_1^{(1)} \tau_2^{(5)}\big)}
 {a_3 a_4 {a_0}^2 \tau_1^{(3)}},\label{eqn:action_T2_tau_8}\\
{T_2}^{-1}\big(\tau_1^{(1)}\big)
 =\frac{a_0 \big(qa_3 \tau_1^{(2)} {T_2}^{-1}\big(\tau_1^{(5)}\big)+a_0 a_1 \tau_2^{(2)} {T_2}^{-1}\big(\tau_2^{(3)}\big)\big)}
 {qa_2 {a_3}^2 a_1 \tau_2^{(5)}},\label{eqn:action_T2_tau_9}\\
{T_2}^{-1}\big(\tau_1^{(2)}\big)
 =\frac{a_0 a_1 \big(a_2a_3 \tau_1^{(2)} {T_2}^{-1}\big(\tau_1^{(5)}\big)+a_0 \tau_2^{(2)} {T_2}^{-1}\big(\tau_2^{(3)}\big)\big)}
 {q^2{a_3}^2 \tau_1^{(5)}},\label{eqn:action_T2_tau_10}\\
{T_2}^{-1}\big(\tau_2^{(3)}\big)
 =\frac{a_2 a_3 \big(a_4 a_0 \tau_1^{(5)} \tau_1^{(2)}+a_2 \tau_1^{(1)} \tau_2^{(5)}\big)}{{a_0}^2 \tau_1^{(4)}},\label{eqn:action_T2_tau_11}\\
{T_2}^{-1}\big(\tau_2^{(4)}\big)
 =\frac{qa_3 a_4 \big(qa_3 \tau_1^{(1)} {T_2}^{-1}\big(\tau_1^{(2)}\big)+a_1 \tau_1^{(2)} {T_2}^{-1}\big(\tau_1^{(1)}\big)\big)}
 {a_0 {a_1}^2 {T_2}^{-1}\big(\tau_2^{(3)}\big)},\label{eqn:action_T2_tau_12}\\
{T_2}^{-1}\big(\tau_2^{(5)}\big)
 =\frac{a_0 a_1 \big(qa_3 \tau_1^{(2)} {T_2}^{-1}\big(\tau_1^{(5)}\big)+a_0 \tau_2^{(2)} {T_2}^{-1}\big(\tau_2^{(3)}\big)\big)}
 {qa_2 {a_3}^2 \tau_1^{(1)}},\label{eqn:action_T2_tau_13}\\
{T_2}^{-1}(\tau_2^{(2)})
 =\frac{a_2 a_3 \big(q^{-1}a_2 {T_2}^{-1}\big(\tau_1^{(1)}\big) {T_2}^{-1}\big(\tau_2^{(5)}\big)+a_0 {T_2}^{-1}\big(\tau_1^{(5)}\big) {T_2}^{-1}\big(\tau_1^{(2)}\big)\big)} {a_4 {a_0}^2 {T_2}^{-1}\big(\tau_2^{(3)}\big)},\label{eqn:action_T2_tau_14}
\end{gather}
\end{subequations}
\vspace{-8mm}
\begin{subequations}\label{eqns:action_T3_tau}
\begin{gather}
T_3\big(\tau_1^{(4)}\big)
 =\frac{q{a_3}^2 a_2 \big(a_1 \tau_1^{(4)} T_3\big(\tau_1^{(1)}\big)+a_3 a_4 \tau_2^{(2)} T_3\big(\tau_2^{(1)}\big)\big)}
 {a_1\tau_1^{(1)}},\label{eqn:action_T3_tau_1}\\
T_3\big(\tau_1^{(5)}\big)
 =\frac{a_3 a_4 \big(qa_3\tau_2^{(2)} T_3\big(\tau_2^{(1)}\big)+a_0 a_1 \tau_1^{(4)} T_3\big(\tau_1^{(1)}\big)\big)}
 {{a_1}^2 \tau_2^{(4)}},\label{eqn:action_T3_tau_2}\\
T_3\big(\tau_1^{(1)}\big)
 =\frac{a_1 a_2 \big(a_3 a_4 \tau_1^{(4)} \tau_1^{(1)}+a_1 \tau_1^{(5)} \tau_2^{(4)}\big)}
 {{a_4}^2 \tau_1^{(3)}},\label{eqn:action_T3_tau_3}\\
T_3\big(\tau_2^{(4)}\big)
 =\frac{a_3 a_4 \big(qa_3 \tau_2^{(2)} T_3\big(\tau_2^{(1)}\big)+a_1 \tau_1^{(4)} T_3\big(\tau_1^{(1)}\big)\big)}
 {a_0 {a_1}^2 \tau_1^{(5)}},\label{eqn:action_T3_tau_4}\\
T_3\big(\tau_2^{(5)}\big)
 =\frac{a_4 a_0 \big(q^{-1}a_4 \tau_1^{(4)} T_3(\tau_2^{(2)})+a_2 \tau_2^{(2)} T_3\big(\tau_1^{(4)}\big)\big)}
 {qa_1 {a_2}^2 T_3\big(\tau_2^{(4)}\big)},\label{eqn:action_T3_tau_5}\\
T_3\big(\tau_2^{(1)}\big)
 =\frac{a_1 \big(a_4 \tau_1^{(4)} \tau_1^{(1)}+a_1 a_2 \tau_1^{(5)} \tau_2^{(4)}\big)}
 {a_3 {a_4}^2 a_2 \tau_1^{(2)}},\label{eqn:action_T3_tau_6}\\
T_3(\tau_2^{(2)})
 =\frac{a_1 a_2 \big(a_4 T_3\big(\tau_1^{(4)}\big) T_3\big(\tau_1^{(1)}\big)+qa_1 T_3\big(\tau_1^{(5)}\big) T_3\big(\tau_2^{(4)}\big)\big)}
 {a_3 {a_4}^2 T_3\big(\tau_2^{(1)}\big)},\label{eqn:action_T3_tau_7}\\
{T_3}^{-1}\big(\tau_1^{(1)}\big)
 =\frac{a_3 \big(a_1 \tau_1^{(1)} \tau_1^{(3)}+a_3 a_4 \tau_1^{(2)} \tau_2^{(1)}\big)}
 {a_4 a_0 {a_1}^2 \tau_1^{(4)}},\label{eqn:action_T3_tau_8}\\
{T_3}^{-1}\big(\tau_1^{(2)}\big)
 =\frac{a_1 \big(qa_4 \tau_1^{(3)} {T_3}^{-1}\big(\tau_1^{(1)}\big)+a_1 a_2 \tau_2^{(3)} {T_3}^{-1}\big(\tau_2^{(4)}\big)\big)}
 {qa_3 {a_4}^2 a_2 \tau_2^{(1)}},\label{eqn:action_T3_tau_9}\\
{T_3}^{-1}\big(\tau_1^{(3)}\big)
 =\frac{a_1 a_2 \big(a_3a_4 \tau_1^{(3)} {T_3}^{-1}\big(\tau_1^{(1)}\big)+a_1 \tau_2^{(3)} {T_3}^{-1}\big(\tau_2^{(4)}\big)\big)}
 {q^2{a_4}^2 \tau_1^{(1)}},\label{eqn:action_T3_tau_10}\\
{T_3}^{-1}\big(\tau_2^{(4)}\big)
 =\frac{a_3 a_4 \big(a_0 a_1 \tau_1^{(1)} \tau_1^{(3)}+a_3 \tau_1^{(2)} \tau_2^{(1)}\big)}{{a_1}^2 \tau_1^{(5)}},\label{eqn:action_T3_tau_11}\\
{T_3}^{-1}\big(\tau_2^{(5)}\big)
 =\frac{qa_4 a_0 \big(qa_4 \tau_1^{(2)} {T_3}^{-1}\big(\tau_1^{(3)}\big)+a_2 \tau_1^{(3)} {T_3}^{-1}\big(\tau_1^{(2)}\big)\big)}
 {a_1 {a_2}^2 {T_3}^{-1}\big(\tau_2^{(4)}\big)},\label{eqn:action_T3_tau_12}\\
{T_3}^{-1}\big(\tau_2^{(1)}\big)
 =\frac{a_1 a_2 \big(qa_4 \tau_1^{(3)} {T_3}^{-1}\big(\tau_1^{(1)}\big)+a_1 \tau_2^{(3)} {T_3}^{-1}\big(\tau_2^{(4)}\big)\big)}
 {qa_3 {a_4}^2 \tau_1^{(2)}},\label{eqn:action_T3_tau_13}\\
{T_3}^{-1}\big(\tau_2^{(3)}\big)
 =\frac{a_3 a_4 \big(q^{-1}a_3 {T_3}^{-1}\big(\tau_1^{(2)}\big) {T_3}^{-1}\big(\tau_2^{(1)}\big)+a_1 {T_3}^{-1}\big(\tau_1^{(1)}\big) {T_3}^{-1}\big(\tau_1^{(3)}\big)\big)}
 {a_0 {a_1}^2 {T_3}^{-1}\big(\tau_2^{(4)}\big)}.\label{eqn:action_T3_tau_14}
\end{gather}
\end{subequations}
Moreover, by using the action of $T_1$, we obtain the following lemma.

\begin{Lemma}
The following discrete Toda type bilinear equations hold
\begin{subequations}\label{eqns:A4_1d_toda}
\begin{gather}
 \tau_{l_1+1}^{l_0,l_2,l_3}\tau_{l_1-1}^{l_0,l_2,l_3}
 = q^{3l_1-l_2-l_3}\frac{a_0a_1}{{a_2}^2a_3}\big({-}1+q^{-l_0+l_1}a_1\big)\big(\tau_{l_1}^{l_0,l_2,l_3}\big)^2\notag\\
 \hphantom{\tau_{l_1+1}^{l_0,l_2,l_3}\tau_{l_1-1}^{l_0,l_2,l_3}=}{} +q^{4(-l_0+l_1)}{a_1}^4\tau_{l_1}^{l_0+1,l_2,l_3}\tau_{l_1}^{l_0-1,l_2,l_3},\label{eqn:A4_1d_toda_T0}\\
 \tau_{l_1+1}^{l_0,l_2,l_3}\tau_{l_1-1}^{l_0,l_2,l_3}
 = q^{-l_0+4l_1-l_2-l_3}\frac{a_0{a_1}^2}{{a_2}^2a_3}\big(1-q^{-l_1+l_2}a_2\big)\big(\tau_{l_1}^{l_0,l_2,l_3}\big)^2\notag\\
\hphantom{\tau_{l_1+1}^{l_0,l_2,l_3}\tau_{l_1-1}^{l_0,l_2,l_3}=}{}
+q^{4(l_1-l_2)}{a_2}^{-4}\tau_{l_1}^{l_0,l_2+1,l_3}\tau_{l_1}^{l_0,l_2-1,l_3},\label{eqn:A4_1d_toda_T2}\\
\tau_{l_1+1}^{l_0,l_2,l_3}\tau_{l_1-1}^{l_0,l_2,l_3}
 =q^{-l_0+3l_1-l_2}\frac{a_1}{{a_2}^2a_3a_4}\big({-}1+q^{l_1-l_3}a_0a_1a_4\big)\big(\tau_{l_1}^{l_0,l_2,l_3}\big)^2\notag\\
\hphantom{\tau_{l_1+1}^{l_0,l_2,l_3}\tau_{l_1-1}^{l_0,l_2,l_3}=}{} +q^{4(l_1-l_3)}{a_0}^4{a_1}^4{a_4}^4\tau_{l_1}^{l_0,l_2,l_3+1}\tau_{l_1}^{l_0,l_2,l_3-1}.\label{eqn:A4_1d_toda_T3}
\end{gather}
\end{subequations}
\end{Lemma}

\begin{proof}
The actions of $T_0$, ${T_1}^{-1}$ and ${T_2}^{-1}$ on $\tau_1^{(1)}$ are given by
\begin{gather}
 T_0\big(\tau_1^{(1)}\big)
 = \frac{q{a_0}^2a_3{a_4}^2\tau_1^{(1)}\big(a_0a_1\tau_1^{(1)}\tau_1^{(3)}+a_3\tau_1^{(2)}\tau_2^{(1)}\big)}{{a_1}^2\tau_1^{(3)}\tau_1^{(5)}}\notag \\
\hphantom{T_0\big(\tau_1^{(1)}\big)=}{} +\frac{q{a_0}^2\tau_2^{(4)}\big(a_1\tau_1^{(1)}\tau_1^{(3)}+a_3a_4\tau_1^{(2)}\tau_2^{(1)}\big)}{a_1\tau_1^{(3)}\tau_1^{(4)}},
 \label{eqn:A4_T0tau11}\\
 {T_1}^{-1}\big(\tau_1^{(1)}\big)
 =\frac{\tau_2^{(1)}\big(a_3a_4\tau_1^{(1)}\tau_1^{(4)}+a_1\tau_1^{(5)}\tau_2^{(4)}\big)}{q{a_2}^2a_3a_4\tau_1^{(3)}\tau_1^{(4)}}\notag\\
\hphantom{{T_1}^{-1}\big(\tau_1^{(1)}\big)=}{}
+\frac{a_1\tau_1^{(1)}\big(a_4\tau_1^{(1)}\tau_1^{(4)}+a_1a_2\tau_1^{(5)}\tau_2^{(4)}\big)}{q{a_2}^3{a_3}^2{a_4}^2\tau_1^{(2)}\tau_1^{(4)}},
 \label{eqn:A4_mT1tau11}\\
 {T_2}^{-1}\big(\tau_1^{(1)}\big)
 = \frac{{a_2}^2\tau_1^{(2)}\big(a_3a_4\tau_1^{(1)}\tau_1^{(4)}+a_1\tau_1^{(5)}\tau_2^{(4)}\big)}{qa_3a_4\tau_1^{(3)}\tau_1^{(4)}}\notag\\
\hphantom{{T_2}^{-1}\big(\tau_1^{(1)}\big)}{}
+\frac{a_1{a_2}^2\tau_1^{(1)}\big(a_4\tau_1^{(1)}\tau_1^{(4)}+a_1\tau_1^{(5)}\tau_2^{(4)}\big)}{q{a_3}^2{a_4}^2\tau_2^{(1)}\tau_1^{(4)}},
 \label{eqn:A4_mT2tau11}
\end{gather}
respectively.
Eliminating the terms $\tau_1^{(3)}$, $\tau_1^{(4)}$, $\tau_2^{(1)}$ and $\tau_2^{(4)}$
from equations \eqref{eqn:A4_T0tau11} and \eqref{eqn:A4_mT1tau11}, we obtain
\begin{gather}\label{eqn:A4_Toda_proof1}
 \tau_1^{(2)}{T_1}^{-1}\big(\tau_1^{(1)}\big)
 =q^{-1}\frac{a_0a_1}{{a_2}^2a_3}\big({-}1+q^{-1}a_1\big)\big(\tau_1^{(1)}\big)^2+q^{-4}{a_1}^4T_0\big(\tau_1^{(1)}\big)\tau_1^{(5)},
\end{gather}
which is equivalent to equation \eqref{eqn:A4_1d_toda_T0}.
Furthermore, eliminating the terms $\tau_1^{(3)}$, $\tau_1^{(4)}$, $\tau_1^{(5)}$ and~$\tau_2^{(4)}$
from equations \eqref{eqn:A4_mT1tau11} and \eqref{eqn:A4_mT2tau11}, we obtain
\begin{gather}\label{eqn:A4_Toda_proof2}
 \tau_1^{(2)}{T_1}^{-1}\big(\tau_1^{(1)}\big)
 =q^{-2}\frac{a_0{a_1}^2}{{a_2}^2a_3}(1-a_2)\big(\tau_1^{(1)}\big)^2+{a_2}^{-4}\tau_2^{(1)}{T_2}^{-1}\big(\tau_1^{(1)}\big),
\end{gather}
which is equivalent to equation \eqref{eqn:A4_1d_toda_T2}.
Eliminating the term $\tau_1^{(2)}{T_1}^{-1}\big(\tau_1^{(1)}\big)$ from equa\-tions \eqref{eqn:A4_Toda_proof1} and~\eqref{eqn:A4_Toda_proof2},
we obtain
\begin{gather}\label{eqn:A4_Toda_proof3}
 T_0\big(\tau_1^{(1)}\big)\tau_1^{(5)}
 =q^2\frac{a_0}{{a_1}^2a_2a_3}(-1+a_4a_0a_3)\big(\tau_1^{(1)}\big)^2+{a_4}^4{a_0}^4{a_3}^4\tau_2^{(1)}{T_2}^{-1}\big(\tau_1^{(1)}\big).
\end{gather}
Applying the transformation $\sigma$ on equation \eqref{eqn:A4_Toda_proof3}, we obtain
\begin{gather*}
 T_1\big(\tau_1^{(2)}\big)\tau_1^{(1)}
 =q^2\frac{a_1}{{a_2}^2a_3a_4}(-1+a_0a_1a_4)\big(\tau_1^{(2)}\big)^2+{a_0}^4{a_1}^4{a_4}^4\tau_2^{(2)}{T_3}^{-1}\tau_1^{(2)},
\end{gather*}
which is equivalent to equation \eqref{eqn:A4_1d_toda_T3}. Therefore, we have completed the proof.
\end{proof}

{\bf Step 2.} In this step, we get the explicit formulae for $\tau_0^{l_0,l_2,l_3}$ and $\tau_1^{l_0,l_2,l_3}$. Letting
\begin{gather}\label{eqn:tau1tau0H}
 \tau_1^{l_0,l_2,l_3}=\tau_0^{l_0,l_2,l_3}H_{l_0,l_2,l_3},
\end{gather}
where
\begin{gather*}
 H_{l_0,l_2,l_3}=H\big(q^{l_0}a_0,q^{l_2}a_2,q^{-l_3}a_4\big),
\end{gather*}
we obtain the following lemma.

\begin{Lemma}
A solution of the system of the equations~\eqref{eqns:A4_conditions_tau} and \eqref{eqns:action_T0_tau}--\eqref{eqns:action_T3_tau}
are given by the solution of the following system under the condition~\eqref{eqn:condition_para}:
\begin{subequations}\label{eqns:tau0tau0tau0tau0}
\begin{gather}
 \tau_0^{0,0,0} \tau_0^{0,1,1}+\frac{a_0 a_4}{q a_2}\,\tau_0^{0,1,0}\tau_0^{0,0,1}=0,\\
 \tau_0^{0,0,0} \tau_0^{1,1,1}-q {a_4}^2 \tau_0^{0,0,1} \tau_0^{1,1,0}=0,\\
 \tau_0^{0,0,0} \tau_0^{1,1,1}-\frac{1}{q {a_2}^2}\,\tau_0^{1,0,1} \tau_0^{0,1,0}=0,\\
 \tau_0^{0,0,0} \tau_0^{1,1,1}-\frac{q}{{a_0}^2}\,\tau_0^{0,1,1} \tau_0^{1,0,0}=0,\\
 \tau_0^{1,0,0} \tau_0^{-1,0,0}-\frac{{a_0}^4 a_4 (1-q{a_0}^{-1})}{q^3a_2}\big(\tau_0^{0,0,0}\big)^2=0,\\
 \tau_0^{0,1,0} \tau_0^{0,-1,0}+\frac{q^2{a_2}^3 a_4 (1-a_2)}{a_0}\big(\tau_0^{0,0,0}\big)^2=0,\\
 \tau_0^{0,0,1} \tau_0^{0,0,-1}-\frac{1-q a_4}{q^3 a_0 a_2 {a_4}^4}\big(\tau_0^{0,0,0}\big)^2=0,
\\
 H_{0,0,0}=q^2a_4 H_{1,1,0}+q(1-qa_4)H_{1,1,1},\tag{H01}\label{eqn:H_type1}\\
 H_{0,0,0}=-q^2a_2a_4 H_{0,1,0}+q^2a_4\big(1-q^{-1}a_0\big)H_{1,1,0},\tag{H02}\label{eqn:H_type2}\\
 H_{0,0,0}=-q^3{a_0}^{-1}a_2a_4 H_{0,1,0}-q^2{a_0}^{-1}\big(1-q^{-1}a_0\big)(1-qa_4)H_{1,1,1},\tag{H03}\label{eqn:H_type3}\\
 H_{0,0,0}=-a_0{a_2}^{-1} H_{1,0,0}-{a_2}^{-1}{a_4}^{-1}(1-qa_4)H_{1,0,1},\tag{H04}\label{eqn:H_type4}\\
 H_{0,0,0}=-q^2{a_0}^{-1}a_2 H_{0,1,0}-q{a_0}^{-1} {a_4}^{-1}(1-qa_4)H_{0,1,1},\tag{H05}\label{eqn:H_type5}\\
 H_{0,0,0}=-qa_0a_4 H_{1,0,0}+q^2a_4(1-a_2) H_{1,1,0},\tag{H06}\label{eqn:H_type6}\\
 H_{0,0,0}=-qa_0{a_2}^{-1}a_4 H_{1,0,0}-q{a_2}^{-1}(1-a_2)(1-qa_4)H_{1,1,1},\tag{H07}\label{eqn:H_type7}\\
 H_{0,0,0}=-{a_2}^{-1}{a_4}^{-1} H_{0,0,1}+q{a_2}^{-1}\big(1-q^{-1}a_0\big)H_{1,0,1},\tag{H08}\label{eqn:H_type8}\\
 H_{0,0,0}=q^2{a_0}^{-1} H_{0,1,1}-q^2{a_0}^{-1}\big(1-q^{-1}a_0\big) H_{1,1,1},\tag{H09}\label{eqn:H_type9}\\
 a_4(1-a_0)H_{2,1,1}H_{0,0,0}=qa_4 H_{1,1,0}H_{1,0,1}-a_0a_4 H_{1,0,0}H_{1,1,1}. \tag{H10}\label{eqn:H_type10}
\end{gather}
\end{subequations}
\end{Lemma}

\begin{proof}
By using notation \eqref{eqn:A4_tau_l_0123} and relation \eqref{eqn:tau1tau0H},
equations \eqref{eqns:A4_conditions_tau} and \eqref{eqns:action_T0_tau}--\eqref{eqns:action_T3_tau}
can be classif\/ied by the type of contiguity relations of function $H_{l_0,l_2,l_3}$ (see Figs.~\ref{fig:type1-2}--\ref{fig:type9-10}) as the following table:
\begin{center}
\begin{tabular}{|l||l|}
\hline
Type & Equation number \\ \hline\hline
Type 1&
\eqref{eqns:A4_conditions_tau_1},
\eqref{eqn:action_T0_tau_4},
\eqref{eqn:action_T0_tau_14},
\eqref{eqn:action_T2_tau_7},
\eqref{eqn:action_T2_tau_13}
\\ \hline
Type 2&
\eqref{eqns:A4_conditions_tau_2},
\eqref{eqn:action_T0_tau_3},
\eqref{eqn:action_T0_tau_5},
\eqref{eqn:action_T0_tau_10},
\eqref{eqn:action_T0_tau_12},
\eqref{eqn:action_T3_tau_2},
\eqref{eqn:action_T3_tau_7},
\eqref{eqn:action_T3_tau_11},
\eqref{eqn:action_T3_tau_13}
\\ \hline
Type 3&
\eqref{eqns:A4_conditions_tau_3},
\eqref{eqn:action_T0_tau_7},
\eqref{eqn:action_T0_tau_13},
\eqref{eqn:action_T2_tau_5},
\eqref{eqn:action_T2_tau_12},
\eqref{eqn:action_T3_tau_4},
\eqref{eqn:action_T3_tau_14}
\\ \hline
Type 4&
\eqref{eqn:action_T0_tau_1},
\eqref{eqn:action_T0_tau_8},
\eqref{eqn:action_T2_tau_4},
\eqref{eqn:action_T2_tau_6},
\eqref{eqn:action_T2_tau_9},
\eqref{eqn:action_T2_tau_14}
\\ \hline
Type 5&
\eqref{eqn:action_T0_tau_2},
\eqref{eqn:action_T0_tau_6},
\eqref{eqn:action_T0_tau_9},
\eqref{eqn:action_T0_tau_11},
\eqref{eqn:action_T2_tau_3},
\eqref{eqn:action_T2_tau_10},
\eqref{eqn:action_T3_tau_1},
\eqref{eqn:action_T3_tau_8}
\\ \hline
Type 6&
\eqref{eqn:action_T2_tau_1},
\eqref{eqn:action_T2_tau_8}
\\ \hline
Type 7&
\eqref{eqn:action_T2_tau_2},
\eqref{eqn:action_T2_tau_11}
\\ \hline
Type 8&
\eqref{eqn:action_T3_tau_3},
\eqref{eqn:action_T3_tau_10}
\\ \hline
Type 9&
\eqref{eqn:action_T3_tau_6},
\eqref{eqn:action_T3_tau_9}
\\ \hline
Type 10&
\eqref{eqn:action_T3_tau_5},
\eqref{eqn:action_T3_tau_12}
\\ \hline\hline
\end{tabular}
\end{center}
Under the condition \eqref{eqn:condition_para}, comparing the coef\/f\/icients of $H_{l_0,l_2,l_3}$ in the same types,
for example \eqref{eqns:A4_conditions_tau_1}$\,\equiv\,$\eqref{eqn:action_T0_tau_4},
and substituting the boundary condition~\eqref{eqn:A4_boundary_tau} in equations \eqref{eqns:A4_1d_toda} with $l_1=0$,
we obtain equations~\eqref{eqns:tau0tau0tau0tau0}. Moreover, by using the relations~\eqref{eqns:tau0tau0tau0tau0},
the equations of Type~1,~\dots, Type~10 are given by equations \eqref{eqn:H_type1}--\eqref{eqn:H_type10}, respectively.
Therefore, we have completed the proof.
\end{proof}

We can easily verify the following lemma by the direct calculation.

\begin{Lemma}
A solution of system~\eqref{eqns:tau0tau0tau0tau0} is given by
\begin{gather*}
 \tau_0^{l_0,l_2,l_3}=\big(q^{l_0}a_0;q,q\big)_\infty\big(q^{l_2+1}a_2;q,q\big)_\infty\big(q^{l_3}{a_4}^{-1};q,q\big)_\infty K_{l_0,l_2,l_3},
\end{gather*}
where
\begin{gather}\label{eqn:tau0_K}
 K_{l_0,l_2,l_3}\\
 =\frac{\big(\Ga\big(q^{l_0+l_2+1}a_0a_2;q,q\big)\Ga\big(q^{l_2+l_3+1}a_2{a_4}^{-1};q,q\big)\Ga\big(q^{l_3+l_0}{a_4}^{-1}a_0;q,q\big)\big)^2}
 {\Ga\big(q^{l_0+l_2+l_3+1}a_0a_2{a_4}^{-1};q,q\big)\big(\Ga\big(q^{l_0+1/3}a_0;q,q\big)
 \Ga\big(q^{l_2+4/3}a_2;q,q\big)\Ga\big(q^{l_3+1/3}{a_4}^{-1};q,q\big)\big)^6}.\nonumber
\end{gather}
\end{Lemma}

\begin{figure}[th!]\centering
\includegraphics[width=0.8\textwidth]{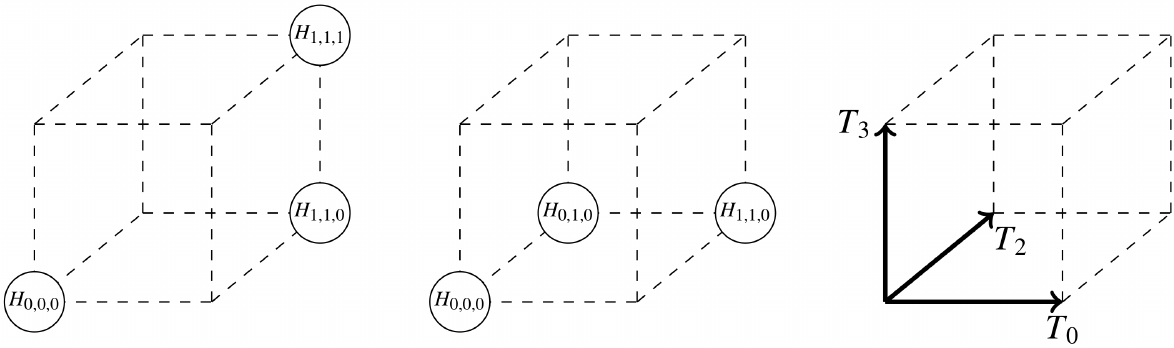}
\caption{Left: Type 1, center: Type 2, right: directions.}\label{fig:type1-2}
\end{figure}

\begin{figure}[ht!]\centering
\includegraphics[width=0.8\textwidth]{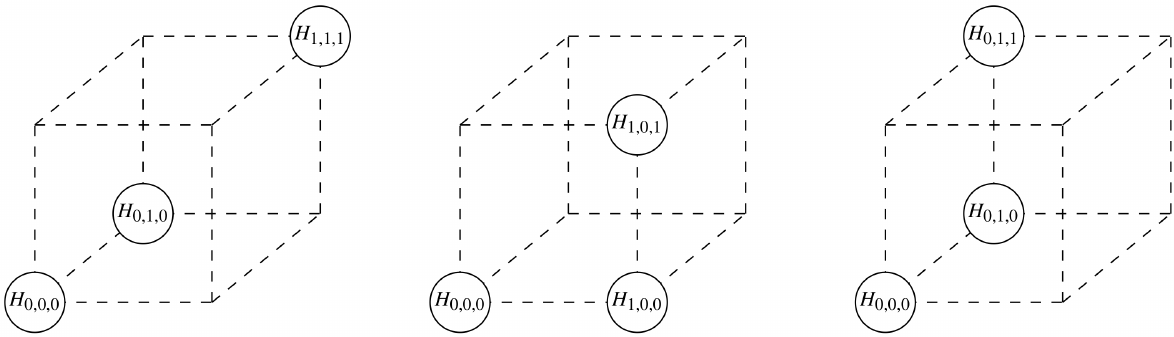}
\caption{Left: Type 3, center: Type 4, right: Type 5.}\label{fig:type3-5}
\end{figure}

\begin{figure}[ht!]\centering
\includegraphics[width=0.8\textwidth]{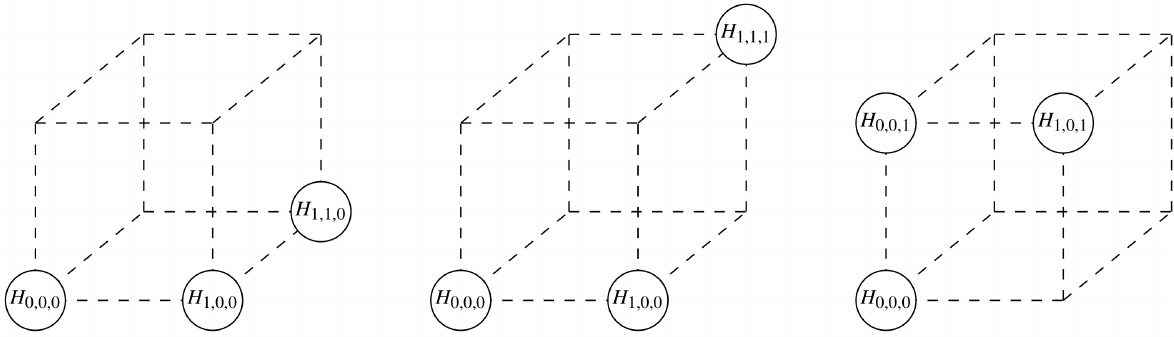}
\caption{Left: Type 6, center: Type 7, right: Type 8.}\label{fig:type6-8}
\end{figure}

\begin{figure}[ht!]\centering
\includegraphics[width=0.7\textwidth]{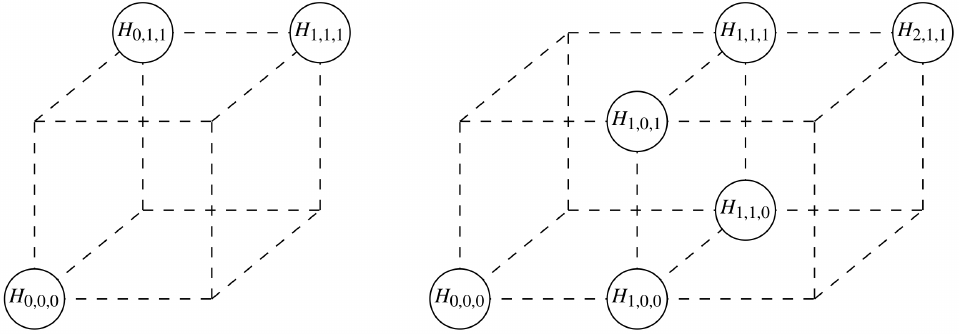}
\caption{Left: Type 9, right: Type 10.}\label{fig:type9-10}
\end{figure}

We consider a solution of system of the equations \eqref{eqn:H_type1}--\eqref{eqn:H_type10}.
First, we get the essential relations for the function $H_{l_0,l_2,l_3}$.

\begin{Lemma}
If the function $H_{l_0,l_2,l_3}$ satisfies equations \eqref{eqn:H_type2}, \eqref{eqn:H_type6} and \eqref{eqn:H_type8} and
the following three-term relation
\begin{gather}\label{eqn:H_3term}
 qa_0a_4(a_0-q)H_{1,0,0}+(a_0-qa_2+qa_0a_2a_4)H_{0,0,0}+qa_2 H_{-1,0,0}=0,\tag{H11}
\end{gather}
then it also satisfies equations \eqref{eqn:H_type1}, \eqref{eqn:H_type3}, \eqref{eqn:H_type4}, \eqref{eqn:H_type5}, \eqref{eqn:H_type7}, \eqref{eqn:H_type9} and \eqref{eqn:H_type10}.
\end{Lemma}

\begin{proof}
Equation \eqref{eqn:H_type4} can be obtained by using equations \eqref{eqn:H_type8} and \eqref{eqn:H_3term} as follows.
Erasing the term $H_{1,0,1}$ from equations \eqref{eqn:H_type8} and \eqref{eqn:H_3term}$_3$, we obtain
\begin{gather}\label{eqn:H_proof_1}
 a_0a_4H_{0,0,0}+(q-a_0a_4)H_{0,0,1}-q H_{-1,0,1}=0.\tag{H12}
\end{gather}
We note that a subscript $i$ of equation number means $T_i$-shifted corresponding equation.
Moreover, erasing the term $H_{0,0,1}$ from equations \eqref{eqn:H_proof_1}$_0$ and \eqref{eqn:H_type8},
we obtain equation \eqref{eqn:H_type4}.
This procedure is described in Fig.~\ref{fig:H_type4}.

In a similar manner, we can derive equations \eqref{eqn:H_type1}, \eqref{eqn:H_type3}, \eqref{eqn:H_type7},
\eqref{eqn:H_type5} and \eqref{eqn:H_type9} as shown in Figs.~\ref{fig:H_type1}--\ref{fig:H_type9}, respectively.
On the other hand, we can prove equation \eqref{eqn:H_type10} by reducing it to equation \eqref{eqn:H_type6}
with equations \eqref{eqn:H_type2}, \eqref{eqn:H_type4} and \eqref{eqn:H_type7}
as shown in Fig.~\ref{fig:H_type10}.
Therefore, we have completed the proof.
\end{proof}

\begin{figure}[th!]\centering
\includegraphics[width=0.8\textwidth]{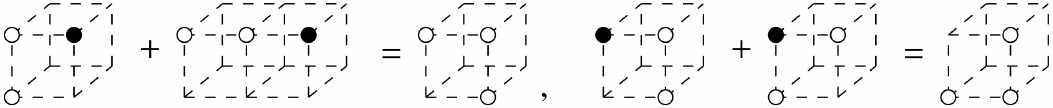}
\caption{Derivation of equation \eqref{eqn:H_type4}. The black points are removed.}
\label{fig:H_type4}
\end{figure}

\begin{figure}[th!]
\begin{minipage}{0.48\hsize}\centering
\includegraphics[width=0.8\textwidth]{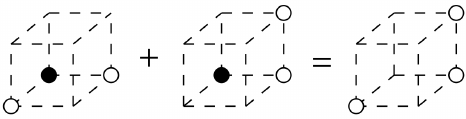}
\caption{Derivation of equation \eqref{eqn:H_type1}.}
\label{fig:H_type1}
\end{minipage}
\quad
\begin{minipage}{0.48\hsize}\centering
\includegraphics[width=0.8\textwidth]{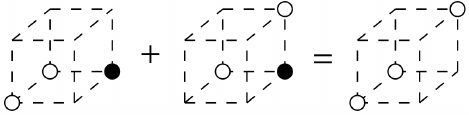}
\caption{Derivation of equation \eqref{eqn:H_type3}.}
\label{fig:H_type3}
\end{minipage}
\end{figure}

\begin{figure}[th!]
\begin{minipage}{0.48\hsize}\centering
\includegraphics[width=0.8\textwidth]{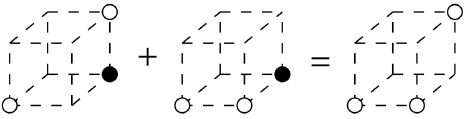}
\caption{Derivation of equation \eqref{eqn:H_type7}.}
\label{fig:H_type7}
\end{minipage}
\quad
\begin{minipage}{0.48\hsize}\centering
\includegraphics[width=0.8\textwidth]{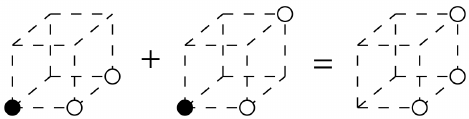}
\caption{Derivation of equation \eqref{eqn:H_type5}.}
\label{fig:H_type5}
\end{minipage}
\end{figure}

\begin{figure}[th!]\centering
\includegraphics[width=0.4\textwidth]{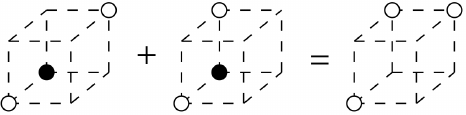}
\caption{Derivation of equation \eqref{eqn:H_type9}.}
\label{fig:H_type9}
\end{figure}

\begin{figure}[th!]\centering
\includegraphics[width=0.8\textwidth]{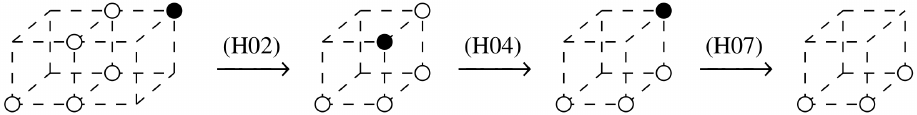}
\caption{Reduction from equation \eqref{eqn:H_type10} to equation \eqref{eqn:H_type6}.}
\label{fig:H_type10}
\end{figure}

Next, we solve the essential relations for the function $H_{l_0,l_2,l_3}$, that is,
equations \eqref{eqn:H_type2}, \eqref{eqn:H_type6}, \eqref{eqn:H_type8} and \eqref{eqn:H_3term}.

\begin{Lemma}
The solution of the system of equations \eqref{eqn:H_type2}, \eqref{eqn:H_type6}, \eqref{eqn:H_type8} and~\eqref{eqn:H_3term} are given by
\begin{gather*}
 H_{l_0,l_2,l_3}=q^{2/3}\left(q^{-l_0-l_2-l_3+2}\frac{a_4}{a_0a_2}\right)^{1/2}\big(q^{l_0-1}{a_0}^{-1};q\big)_\infty G_{l_0,l_2,l_3},
\end{gather*}
where
\begin{gather}
 G_{l_0,l_2,l_3} = C_1 q^{(-2l_2+l_3)/2}{a_2}^{-1}{a_4}^{-1/2}
 \frac{\big(q^{-l_2+l_3}{a_2}^{-1}{a_4}^{-1};q\big)_\infty}{\big(q^{-l_2+1}{a_2}^{-1};q\big)_\infty}
 \frac{\Theta\big(q^{-l_0+l_2-1/2}{a_0}^{-1}a_2;q\big)}{\Theta\big(q^{-l_0}{a_0}^{-1};q\big)}\notag\\
\hphantom{G_{l_0,l_2,l_3} =+{}}{} \times
 {}_1\varphi_1\left(\begin{matrix}q^{-l_2+1}{a_2}^{-1}\\q^{-l_2+l_3}{a_2}^{-1}{a_4}^{-1}\end{matrix};q,q^{-l_0+l_3}{a_0}^{-1}{a_4}^{-1}\right)\notag\\
\hphantom{G_{l_0,l_2,l_3} =}{}+C_2 q^{(l_2-2l_3)/2}{a_2}^{1/2}a_4
 \frac{\big(q^{l_2-l_3+2}a_2a_4;q\big)_\infty}{\big(q^{-l_3+2}a_4;q\big)_\infty}
 \frac{\Theta\big(q^{-l_0+l_3-3/2}{a_0}^{-1}{a_4}^{-1};q\big)}{\Theta\big(q^{-l_0}{a_0}^{-1};q\big)}\notag\\
\hphantom{G_{l_0,l_2,l_3} =+{}}{}\times
 {}_1\varphi_1\left(\begin{matrix}q^{-l_3+2}a_4\\q^{l_2-l_3+2}a_2a_4\end{matrix};q,q^{-l_0+l_2+1}{a_0}^{-1}a_2\right).\label{eqn:tau1_G}
\end{gather}
Here, $C_i=C_i(l_0,l_2,l_3)$, $i=1,2$, are periodic functions of period one for $l_0,l_2,l_3\in\bbZ$, i.e.,
\begin{gather*}
 C_i(l_0+1,l_2,l_3)=C_i(l_0,l_2+1,l_3)=C_i(l_0,l_2,l_3+1)=C_i(l_0,l_2,l_3).
\end{gather*}
\end{Lemma}

\begin{proof}
By letting
\begin{gather*}
 H_{l_0,l_2,l_3}=q^{2/3}\left(q^{-l_0-l_2-l_3+2}\frac{t}{\al\be}\right)^{1/2}\big(q^{l_0-1}t^{-1};q\big)_\infty G\big(q^{-l_0}t,q^{l_2}\al,q^{l_3}\be\big),
\end{gather*}
where
\begin{gather*}
 t={a_0}^{-1},\qquad \al=a_2,\qquad \be={a_4}^{-1},
\end{gather*}
equations \eqref{eqn:H_type2}, \eqref{eqn:H_type6}, \eqref{eqn:H_type8} and \eqref{eqn:H_3term} can be rewritten as the following
\begin{gather}
 \be G(q t,\al ,\be)-q G(t,q\al,\be)+q^{3/2}\al G(q t,q \al,\be)=0, \label{eqn:proof_G_type2}\\
 \be \big(1-q^2 t\big) G(q t,\al,\be)+q^3(1-\al) t G(t,q\al,\be)-q^{3/2} G(t,\al,\be)=0, \label{eqn:proof_G_type6}\\
 q^{1/2}\al G(q t,\al,\be)-q^{1/2} G(t,\al,q\be)+\be G(q t,\al,q\be)=0, \label{eqn:proof_G_type8}\\
 \al \be \big(q^3 t-1\big) G\big(q^2 t,\al,\be\big)-\big(q^{5/2}\al \be t-q^{1/2} (q\al+\be)\big) G(q t,\al,\be)-q^2 G(t,\al,\be)=0, \label{eqn:proof_G_3term}
\end{gather}
respectively. Substituting
\begin{gather*}
 G(t,\al,\be)=\sum_{n=0}^\infty c_nt^{n+\rho},
\end{gather*}
where $c_n=c_n(\al,\be)$,
in equation \eqref{eqn:proof_G_3term}, we obtain
\begin{gather}
 G(t,\al,\be) =A(\al,\be)\frac{\Theta\big(q^{-1/2}\al t;q\big)}{\Theta(t;q)}\,{}_1\varphi_1\left(\begin{matrix}q\al^{-1}\\\al^{-1}\be\end{matrix};q,\be t\right)\notag\\
\hphantom{G(t,\al,\be) =}{} +B(\al,\be)\frac{\Theta\big(q^{-3/2}\be t;q\big)}{\Theta(t;q)}\,{}_1\varphi_1\left(\begin{matrix}q^2\be^{-1}\\q^2\al\be^{-1}\end{matrix};q,q \al t\right),\label{eqn:proof_G_solution}
\end{gather}
where $A(\al,\be)$ and $B(\al,\be)$ are arbitrary functions.
Moreover, substituting \eqref{eqn:proof_G_solution} in
equations \eqref{eqn:proof_G_type2}, \eqref{eqn:proof_G_type6} and \eqref{eqn:proof_G_type8}, we obtain the following relations
\begin{subequations}\label{eqns:rels_AB}
\begin{alignat}{3}
 &A(q\al,\be)=\frac{1-q^{-1}\al^{-1}\be}{q\big(1-\al^{-1}\big)}\,A(\al,\be),\qquad&&
 A(\al,q\be)=\frac{q^{1/2}}{1-\al^{-1}\be}\,A(\al,\be),&\\
 &B(q\al,\be)=\frac{q^{1/2}}{1-q^2\al\be^{-1}} B(\al,\be),\qquad&&
 B(\al,q\be)=\frac{1-q\al\be^{-1}}{q \big(1-q\be^{-1}\big)} B(\al,\be),&
\end{alignat}
\end{subequations}
which can be solved by
\begin{gather*}
 A(\al,\be)=\al^{-1}\be^{1/2}\frac{\big(\al^{-1}\be;q\big)_\infty}{\big(q\al^{-1};q\big)_\infty},\qquad
 B(\al,\be)=\al^{1/2}\be^{-1}\frac{\big(q^2\al\be^{-1};q\big)_\infty}{\big(q^2\be^{-1};q\big)_\infty}.
\end{gather*}
To obtain the relations \eqref{eqns:rels_AB}, we used the following recurrence relations of hypergeometric series~${}_1\varphi_1$:
\begin{gather*}
 {}_1\varphi_1\left(\begin{matrix}a\\b\end{matrix};q,z\right)
 -\frac{q+b-q^2 z}{q}\,{}_1\varphi_1\left(\begin{matrix}a\\b\end{matrix};q,qz\right)
 -\frac{q^2 a z-b}{q}\,{}_1\varphi_1\left(\begin{matrix}a\\b\end{matrix};q,q^2z\right)=0,\\
{}_1\varphi_1\left(\begin{matrix}a\\b\end{matrix};q,z\right)
 ={}_1\varphi_1\left(\begin{matrix}a\\b\end{matrix};q,qz\right)
 +\frac{(a-1) z}{1-b}\,{}_1\varphi_1\left(\begin{matrix}qa\\qb\end{matrix};q,qz\right),\\
{}_1\varphi_1\left(\begin{matrix}a\\b\end{matrix};q,z\right)
 =\frac{1}{1-b}\,{}_1\varphi_1\left(\begin{matrix}a\\qb\end{matrix};q,z\right)
 -\frac{b}{1-b}\,{}_1\varphi_1\left(\begin{matrix}a\\qb\end{matrix};q,qz\right),
\end{gather*}
which can be verif\/ied by the direct calculation. Therefore, we have completed the proof.
\end{proof}

{\bf Step 3.} In this f\/inal step, we give the hypergeometric $\tau$ functions of $\widetilde{W}\big(A_4^{(1)}\big)$-type.

Substituting
\begin{gather*}
 \tau_{l_1}^{l_0,l_2,l_3}
 =q^{2{l_1}^3/3}\left(q^{-l_0-l_2-l_3+2}\frac{a_4}{a_0a_2}\right)^{{l_1}^2/2}
 \frac{(q^{l_0-l_1}{a_0}^{-1};q,q)_\infty}{(q^{l_0}{a_0}^{-1};q,q)_\infty}
 \tau_0^{l_0,l_2,l_3} \Phi_{l_1}^{l_0,l_2,l_3},
\end{gather*}
in equation \eqref{eqn:A4_1d_toda_T0}, we obtain the following bilinear equation
\begin{gather}\label{eqn:Phi_bilinear}
 \Phi_{l_1+1}^{l_0,l_2,l_3}\Phi_{l_1-1}^{l_0,l_2,l_3}
 =(\Phi_{l_1}^{l_0,l_2,l_3})^2-\Phi_{l_1}^{l_0+1,l_2,l_3}\Phi_{l_1}^{l_0-1,l_2,l_3}.
\end{gather}
In general, equation \eqref{eqn:Phi_bilinear} admits a solution expressed in terms of Jacobi--Trudi type determinant
\begin{gather*}
 \Phi_{l_1}^{l_0,l_2,l_3}=\det(c_{l_0+i-j,l_2,l_3})_{i,j=1,\dots,l_1},
\end{gather*}
where $l_1\in\bbZ_{>1}$, under the boundary conditions
\begin{gather*}
 \Phi_{l_1}^{l_0,l_2,l_3}=0\qquad(l_1<0),\qquad
 \Phi_0^{l_0,l_2,l_3}=1,\qquad
 \Phi_1^{l_0,l_2,l_3}=c_{l_0,l_2,l_3},
\end{gather*}
where $c_{l_0,l_2,l_3}$ is an arbitrary function.
Therefore, we obtain the following theorem.

\begin{Theorem}\label{theorem:A4_hypertau}
The hypergeometric $\tau$ functions of $\widetilde{W}(A_4^{(1)})$-type are given by the following
\begin{gather*}
 \tau_0^{l_0,l_2,l_3}=\big(q^{l_0}a_0;q,q\big)_\infty\big(q^{l_2+1}a_2;q,q\big)_\infty\big(q^{l_3}{a_4}^{-1};q,q\big)_\infty K_{l_0,l_2,l_3},\\
 \tau_{l_1}^{l_0,l_2,l_3} =q^{2{l_1}^3/3}\left(q^{-l_0-l_2-l_3+2}\frac{a_4}{a_0a_2}\right)^{{l_1}^2/2}
 \frac{\big(q^{l_0-l_1}{a_0}^{-1};q,q\big)_\infty}{\big(q^{l_0}{a_0}^{-1};q,q\big)_\infty}
 \tau_0^{l_0,l_2,l_3} \Phi_{l_1}^{l_0,l_2,l_3},
\end{gather*}
where $l_0,l_2,l_3\in\bbZ$, $l_1\in\bbZ_{>0}$ and the functions $\big\{\Phi_{l_1}^{l_0,l_2,l_3}\big\}_{l_i\in\bbZ,\, l_1>0}$ are given by the following $l_1\times l_1$ determinants
\begin{gather}\label{eqn:def_Phi}
 \Phi_{l_1}^{l_0,l_2,l_3}
 =\begin{vmatrix}
 G_{l_0,l_2,l_3}&G_{l_0+1,l_2,l_3}&\cdots&G_{l_0+l_1-1,l_2,l_3}\\
 G_{l_0-1,l_2,l_3}&G_{l_0,l_2,l_3}&\cdots&G_{l_0+l_1-2,l_2,l_3}\\
 \vdots&\vdots&\cdots&\vdots\\
 G_{l_0-l_1+1,l_2,l_3}&G_{l_0-l_1+2,l_2,l_3}&\cdots&G_{l_0,l_2,l_3}
 \end{vmatrix}.
\end{gather}
Here, the functions $K_{l_0,l_2,l_3}$ and $G_{l_0,l_2,l_3}$ are given in equations~\eqref{eqn:tau0_K} and~\eqref{eqn:tau1_G}, respectively.
\end{Theorem}

\subsection{Discrete Painlev\'e equations}
Let us def\/ine the ten $f$-variables by
\begin{gather}\label{eqn:A4_def_f}
 f_1^{(j)}=\frac{\tau_1^{(j+1)}\tau_2^{(j)}}{\tau_1^{(j)}\tau_1^{(j+2)}},\qquad
 f_2^{(j)}=\frac{a_ja_{j+1}}{{a_{j+3}}^2}\,
 \frac{\tau_1^{(j+1)}\big(a_{j+2}a_{j+3}\tau_1^{(j)}\tau_1^{(j+3)}+a_j\tau_1^{(j+4)}\tau_2^{(j+3)}\big)}{\tau_1^{(j)}\tau_1^{(j+2)}\tau_2^{(j+1)}},
\end{gather}
where $j\in\bbZ/5\bbZ$. From the def\/inition above and conditions \eqref{eqns:A4_conditions_tau}, the following eight relations hold
\begin{gather*}
 f_2^{(j)}=\frac{a_ja_{j+1}\big(a_{j+2}a_{j+3}+a_jf_1^{(j+3)}\big)}{{a_{j+3}}^2f_1^{(j+1)}},\qquad j\in\bbZ/5\bbZ,\\
 a_4{a_0}^2f_1^{(2)}f_1^{(3)}=a_2a_3\big(a_0+a_2f_1^{(5)}\big),\qquad
 a_0{a_1}^2f_1^{(3)}f_1^{(4)}=a_3a_4\big(a_1+a_3f_1^{(1)}\big),\\
 a_2{a_3}^2f_1^{(5)}f_1^{(1)}=a_0a_1\big(a_3+a_0f_1^{(3)}\big).
\end{gather*}
Therefore, the $f$-variables are essentially two. The action of $\widetilde{W}\big(A_4^{(1)}\big)$ on these variables $f_i^{(j)}$ is given by the following lemma, which follows from the actions \eqref{eqns:A4_weylaction_tau}.

\begin{Lemma}
The action of $\widetilde{W}\big(A_4^{(1)}\big)$ on variables $f_i^{(j)}$ is given by
\begin{gather*}
 s_j\big(f_1^{(j+3)}\big)=f_2^{(j+3)},\qquad s_j\big(f_2^{(j+3)}\big)=f_1^{(j+3)},\\
 s_j\big(f_1^{(j)}\big)=\frac{a_{j+4}}{a_ja_{j+1}{a_{j+2}}^2} \frac{a_{j+2}+a_ja_{j+4}f_1^{(j+2)}}{f_1^{(j+4)}},\\
 s_j\big(f_2^{(j+2)}\big)=\frac{a_ja_{j+3}a_{j+4}}{a_{j+1}} \frac{a_{j+2}+a_ja_{j+4}f_1^{(j+2)}+a_ja_{j+1}a_{j+2}f_1^{(j+4)}}{f_1^{(j+4)}f_2^{(j+3)}},\\
 s_j\big(f_2^{(j+4)}\big)=\frac{a_ja_{j+1}{a_{j+2}}^2}{a_{j+4}}\frac{f_1^{(j+4)}f_1^{(j)}f_2^{(j+4)}}{a_{j+2}+a_ja_{j+4}f_1^{(j+2)}},\\
 s_j\big(f_2^{(j)}\big)=\frac{a_ja_{j+1}a_{j+4}+a_{j+3}a_{j+4}f_1^{(j+1)}+a_j{a_{j+1}}^2a_{j+2}f_1^{(j+4)}}{a_ja_{j+1}a_{j+3}f_1^{(j+1)}f_1^{(j+4)}},\\
 \pi\big(f_i^{(j)}\big)=f_i^{(j+1)},\qquad \iota\big(f_1^{(j)}\big)=f_1^{(3-j)},\qquad
 \iota\big(f_2^{(j)}\big)=\frac{a_{2-j}\big(a_{5-j}+a_{2-j}a_{3-j}f_1^{(5-j)}\big)}{a_{3-j}a_{4-j}{a_{5-j}}^2f_1^{(2-j)}},
\end{gather*}
where $i=1,2$ and $j\in\bbZ/5\bbZ$.
\end{Lemma}

It is well known that the translation part of $\widetilde{W}\big(A_4^{(1)}\big)$ give discrete Painlev\'e equations~\cite{SakaiH2001:MR1882403}. Let
\begin{gather}\label{eqn:A4_def_XY}
 X_{l_1}^{(i)}={T_1}^{l_1}\big(f_1^{(i)}\big),\qquad
 Y_{l_1}^{(i)}={T_1}^{l_1}\big(f_2^{(i)}\big),\qquad
 \al_{l_1}^{(i)}={T_1}^{l_1}(a_i)=
 \begin{cases}
 q^{l_1}a_1&\text{if} \ i=1,\\
 q^{-l_1}a_2&\text{if} \ i=2,\\
 a_i&\text{otherwise}.
 \end{cases}
\end{gather}
The action of $T_i$:
\begin{alignat*}{3}
 &T_i\colon \ && (a_i,a_{i+1})\mapsto\big(qa_i,q^{-1}a_{i+1}\big),&\\
 &&& T_i\big(f_1^{(i+3)}\big)f_1^{(i+3)} =\frac{a_{i+3}}{{a_i}^2{a_{i+1}}^2a_{i+4}}
 \frac{\big(a_{i+1}+a_{i+3}a_{i+4}f_1^{(i+1)}\big)\big(a_{i+1}+a_{i+3}f_1^{(i+1)}\big)}{a_ia_{i+1}+a_{i+3}f_1^{(i+1)}},& \\
 &&&{T_i}^{-1}\big(f_1^{(i+1)}\big)f_1^{(i+1)}
 =\frac{a_i{a_{i+1}}^3}{{a_{i+3}}^2} \frac{\big(a_{i+2}a_{i+3}+a_if_1^{(i+3)}\big)\big(a_{i+3}+a_if_1^{(i+3)}\big)}{a_{i+3}+a_ia_{i+1}f_1^{(i+3)}},
\end{alignat*}
where $i\in\bbZ/5\bbZ$, lead a $q$-discrete analogue of Painlev\'e~V equation~\cite{SakaiH2001:MR1882403}
\begin{gather}
 T_i\big(X_{l_1}^{(i+3)}\big)X_{l_1}^{(i+3)}
 =\frac{\al_{l_1}^{(i+3)}}{\big(\al_{l_1}^{(i)}\big)^2\big(\al_{l_1}^{(i+1)}\big)^2\al_{l_1}^{(i+4)}}\nonumber\\
 \hphantom{T_i\big(X_{l_1}^{(i+3)}\big)X_{l_1}^{(i+3)}=}{}\times
 \frac{\big(\al_{l_1}^{(i+1)}+\al_{l_1}^{(i+3)}\al_{l_1}^{(i+4)}X_{l_1}^{(i+1)}\big)\big(\al_{l_1}^{(i+1)}
 +\al_{l_1}^{(i+3)}X_{l_1}^{(i+1)}\big)}{\al_{l_1}^{(i)}\al_{l_1}^{(i+1)}+\al_{l_1}^{(i+3)}X_{l_1}^{(i+1)}},\nonumber\\
 {T_i}^{-1}\big(X_{l_1}^{(i+1)}\big)X_{l_1}^{(i+1)}
 =\frac{\al_{l_1}^{(i)}\big(\al_{l_1}^{(i+1)}\big)^3}{\big(\al_{l_1}^{(i+3)}\big)^2}
 \frac{\big(\al_{l_1}^{(i+2)}\al_{l_1}^{(i+3)}\!+\al_{l_1}^{(i)} X_{l_1}^{(i+3)}\big)\big(\al_{l_1}^{(i+3)}\!+\al_{l_1}^{(i)} X_{l_1}^{(i+3)}\big)}{\al_{l_1}^{(i+3)}\!+\al_{l_1}^{(i)}\al_{l_1}^{(i+1)}X_{l_1}^{(i+3)}}.\!\!\!\!\!\label{eqn:qp5_1}
\end{gather}
Moreover, the action of $T_{23}^{(i)}=T_{i+2}T_{i+3}$:
\begin{alignat*}{3}
 & T_{23}^{(i)}\colon \ && (a_{i+2},a_{i+4})\mapsto\big(qa_{i+2},q^{-1}a_{i+4}\big),&\\
 &&& \left(T_{23}^{(i)}\big(f_2^{(i+2)}\big)f_1^{(i+3)}-\frac{a_{i+2}a_{i+3}a_{i+4}}{a_i}\right)
 \left(f_2^{(i+2)}f_1^{(i+3)}-\frac{a_{i+2}a_{i+3}a_{i+4}}{a_i}\right) &\\
 &&& \quad{}=\frac{{a_{i+2}}^3a_{i+3}a_{i+4}}{{a_i}^2}
 \frac{\big(a_{i+3}+a_if_1^{(i+3)}\big)\big(a_{i+3}+a_ia_{i+1}f_1^{(i+3)}\big)}{a_{i+2}a_{i+3}+a_if_1^{(i+3)}},& \\
 &&& \left(f_2^{(i+2)}f_1^{(i+3)}-\frac{a_{i+2}a_{i+3}a_{i+4}}{a_i}\right)
 \left(f_2^{(i+2)}{T_{23}^{(i)}}^{-1}\big(f_1^{(i+3)}\big)-\frac{a_{i+2}a_{i+3}a_{i+4}}{a_i}\right)& \\
 &&& \quad{} =\frac{a_{i+2}a_{i+3}}{{a_i}^2a_{i+1}a_{i+4}}
 \frac{\big(a_{i+1}a_{i+2}a_{i+4}+f_2^{(i+2)}\big)\big(a_{i+2}a_{i+4}+f_2^{(i+2)}\big)}{a_{i+2}+a_if_2^{(i+2)}},
\end{alignat*}
where $i\in\bbZ/5\bbZ$, and that of $T_{13}^{(i)}=T_{i+1}T_{i+3}$:
\begin{alignat*}{3}
 &T_{13}^{(i)}\colon \ && (a_{i+1},a_{i+2},a_{i+3},a_{i+4})\mapsto\big(qa_{i+1},q^{-1}a_{i+2},qa_{i+3},q^{-1}a_{i+4}\big),& \\
 &&& \left(T_{13}^{(i)}\big(f_1^{(i+1)}\big)f_1^{(i+2)}-\frac{a_{i+1} a_{i+2}}{a_{i+3} a_{i+4}}\right)
 \left(f_1^{(i+1)} f_1^{(i+2)}-\frac{a_{i+1} a_{i+2}}{a_{i+3} a_{i+4}}\right)& \\
 &&& \quad{} =\frac{{a_{i+1}}^3 a_{i+2}}{a_{i+3} {a_{i+4}}^2}
 \frac{\big(a_{i+2}+a_i a_{i+4} f_1^{(i+2)}\big) \big(a_{i+2}+a_{i+4} f_1^{(i+2)}\big)}{a_{i+1} a_{i+2}+a_{i+4} f_1^{(i+2)}},& \\
 &&& \left(f_1^{(i+1)} f_1^{(i+2)}-\frac{a_{i+1} a_{i+2}}{a_{i+3} a_{i+4}}\right)
 \left(f_1^{(i+1)}{T_{13}^{(i)}}^{-1}\big(f_1^{(i+2)}\big)-\frac{a_{i+1} a_{i+2}}{a_{i+3} a_{i+4}}\right)& \\
 &&& \quad{} =\frac{a_{i+1} a_{i+2}}{a_i {a_{i+3}}^2 {a_{i+4}}^2}
 \frac{\big(a_{i+1}+a_{i+3} f_1^{(i+1)}\big)\big(a_i a_{i+1}+a_{i+3} f_1^{(i+1)}\big)}{a_{i+1}+a_{i+3} a_{i+4} f_1^{(i+1)}},&
\end{alignat*}
where $i\in\bbZ/5\bbZ$, respectively give the systems
\begin{subequations}\label{eqn:qp5_2}
\begin{gather}
\left(T_{23}^{(i)}\big(Y_{l_1}^{(i+2)}\big)X_{l_1}^{(i+3)}-\frac{\al_{l_1}^{(i+2)}\al_{l_1}^{(i+3)}\al_{l_1}^{(i+4)}}{\al_{l_1}^{(i)}}\right)
 \left(Y_{l_1}^{(i+2)}X_{l_1}^{(i+3)}-\frac{\al_{l_1}^{(i+2)}\al_{l_1}^{(i+3)}\al_{l_1}^{(i+4)}}{\al_{l_1}^{(i)}}\right)\nonumber\\
 \qquad{}=\frac{{\al_{l_1}^{(i+2)}}^3\al_{l_1}^{(i+3)}\al_{l_1}^{(i+4)}}{\big(\al_{l_1}^{(i)}\big)^2}
 \frac{\big(\al_{l_1}^{(i+3)}+\al_{l_1}^{(i)}X_{l_1}^{(i+3)}\big)\big(\al_{l_1}^{(i+3)}+\al_{l_1}^{(i)}\al_{l_1}^{(i+1)}
 X_{l_1}^{(i+3)}\big)}{\al_{l_1}^{(i+2)}\al_{l_1}^{(i+3)}+\al_{l_1}^{(i)}X_{l_1}^{(i+3)}},\label{eqn:qp5_2-a}\\
\left(Y_{l_1}^{(i+2)}X_{l_1}^{(i+3)}-\frac{\al_{l_1}^{(i+2)}\al_{l_1}^{(i+3)}\al_{l_1}^{(i+4)}}{\al_{l_1}^{(i)}}\right)
 \left(Y_{l_1}^{(i+2)}{T_{23}^{(i)}}^{-1}\big(X_{l_1}^{(i+3)}\big)-\frac{\al_{l_1}^{(i+2)}\al_{l_1}^{(i+3)}\al_{l_1}^{(i+4)}}{\al_{l_1}^{(i)}}\right)\nonumber\\
 \qquad{}=\frac{\al_{l_1}^{(i+2)}\al_{l_1}^{(i+3)}}{\big(\al_{l_1}^{(i)}\big)^2\al_{l_1}^{(i+1)}\al_{l_1}^{(i+4)}}
 \frac{\big(\al_{l_1}^{(i+1)}\al_{l_1}^{(i+2)}\al_{l_1}^{(i+4)}+Y_{l_1}^{(i+2)}\big)
 \big(\al_{l_1}^{(i+2)}\al_{l_1}^{(i+4)}+Y_{l_1}^{(i+2)}\big)}{\al_{l_1}^{(i+2)}+\al_{l_1}^{(i)}Y_{l_1}^{(i+2)}},\label{eqn:qp5_2-b}
 \end{gather}
\end{subequations}
and
\begin{subequations}\label{eqn:qp5_3}
\begin{gather}
\left(T_{13}^{(i)}\big(X_{l_1}^{(i+1)}\big)X_{l_1}^{(i+2)}-\frac{\al_{l_1}^{(i+1)} \al_{l_1}^{(i+2)}}{\al_{l_1}^{(i+3)} \al_{l_1}^{(i+4)}}\right)
 \left(X_{l_1}^{(i+1)} X_{l_1}^{(i+2)}-\frac{\al_{l_1}^{(i+1)} \al_{l_1}^{(i+2)}}{\al_{l_1}^{(i+3)} \al_{l_1}^{(i+4)}}\right)\nonumber\\
 \qquad{}=\frac{\big(\al_{l_1}^{(i+1)}\big)^3 \al_{l_1}^{(i+2)}}{\al_{l_1}^{(i+3)} \big(\al_{l_1}^{(i+4)}\big)^2}
 \frac{\big(\al_{l_1}^{(i+2)}+\al_{l_1}^{(i)} \al_{l_1}^{(i+4)} X_{l_1}^{(i+2)}\big) \big(\al_{l_1}^{(i+2)}+\al_{l_1}^{(i+4)} X_{l_1}^{(i+2)}\big)}{\al_{l_1}^{(i+1)} \al_{l_1}^{(i+2)}+\al_{l_1}^{(i+4)} X_{l_1}^{(i+2)}},\label{eqn:qp5_3-a}\\
 \left(X_{l_1}^{(i+1)} X_{l_1}^{(i+2)}-\frac{\al_{l_1}^{(i+1)} \al_{l_1}^{(i+2)}}{\al_{l_1}^{(i+3)} \al_{l_1}^{(i+4)}}\right)
 \left(X_{l_1}^{(i+1)}{T_{13}^{(i)}}^{-1}\big(X_{l_1}^{(i+2)}\big)-\frac{\al_{l_1}^{(i+1)} \al_{l_1}^{(i+2)}}{\al_{l_1}^{(i+3)} \al_{l_1}^{(i+4)}}\right)\nonumber\\
 \qquad{}=\frac{\al_{l_1}^{(i+1)} \al_{l_1}^{(i+2)}}{\al_{l_1}^{(i)} \big(\al_{l_1}^{(i+3)}\big)^2 \big(\al_{l_1}^{(i+4)}\big)^2}
 \frac{\big(\al_{l_1}^{(i+1)}+\al_{l_1}^{(i+3)} X_{l_1}^{(i+1)}\big)\big(\al_{l_1}^{(i)} \al_{l_1}^{(i+1)}+\al_{l_1}^{(i+3)} X_{l_1}^{(i+1)}\big)}{\al_{l_1}^{(i+1)}+\al_{l_1}^{(i+3)} \al_{l_1}^{(i+4)} X_{l_1}^{(i+1)}}.\label{eqn:qp5_3-b}
\end{gather}
\end{subequations}
Systems \eqref{eqn:qp5_2} and \eqref{eqn:qp5_3} are also known as $q$-discrete analogues of Painlev\'e~V equation~\cite{TGCR2004:MR2058894}.

From equation \eqref{eqns:A4_config_tau}, def\/initions \eqref{eqn:A4_def_f} and~\eqref{eqn:A4_def_XY}
and Theorem~\ref{theorem:A4_hypertau}, we obtain the following corollary.

\begin{Corollary}\label{corollary:A4_HGsol}
Under the condition \eqref{eqn:condition_para},
the hypergeometric solutions of $q$-Painlev\'e equations \eqref{eqn:qp5_1}, \eqref{eqn:qp5_2} and \eqref{eqn:qp5_3} are given by
\begin{gather*}
 X_{l_1}^{(1)} =q^{l_1+1/2} \frac{\Phi_{l_1}^{1,1,1} \Phi_{l_1+1}^{1,0,1}}{\Phi_{l_1}^{1,0,1} \Phi_{l_1+1}^{1,1,1}},\qquad
 X_{l_1}^{(2)} =-\frac{q a_2}{a_0 a_4} \frac{\Phi_{l_1+1}^{1,0,2} \Phi_{l_1+1}^{1,1,1}}{\Phi_{l_1+1}^{1,0,1} \Phi_{l_1+1}^{1,1,2}},\\
 X_{l_1}^{(3)} =\frac{1-a_4}{q^{l_1+1/2} a_0 a_2 {a_4}^2} \frac{\Phi_{l_1}^{0,0,0} \Phi_{l_1+1}^{1,1,2}}{\Phi_{l_1}^{0,0,1} \Phi_{l_1+1}^{1,1,1}},\qquad
 X_{l_1}^{(4)} =-\frac{a_0 a_4}{q^{l_1+1/2}a_2} \frac{\Phi_{l_1}^{0,0,1} \Phi_{l_1+1}^{2,1,2}}{\Phi_{l_1}^{1,0,1} \Phi_{l_1+1}^{1,1,2}},\\
 X_{l_1}^{(5)} =\frac{q^{l_1+1}-a_0}{q^{1/2}} \frac{\Phi_{l_1}^{1,0,1} \Phi_{l_1+1}^{0,0,1}}{\Phi_{l_1}^{0,0,1} \Phi_{l_1+1}^{1,0,1}},\qquad
 Y_{l_1}^{(1)} =\frac{1}{q^{1/2} a_2} \frac{\Phi_{l_1+1}^{1,0,1} \Phi_{l_1+1}^{2,1,2}}{\Phi_{l_1+1}^{1,0,2} \Phi_{l_1+1}^{1,1,1}}
 \left( \frac{\Phi_{l_1}^{0,0,1}}{\Phi_{l_1}^{1,0,1}} -\frac{q^{1/2}}{a_4}\frac{\Phi_{l_1+1}^{1,1,2}}{\Phi_{l_1+1}^{2,1,2}} \right),\\
 Y_{l_1}^{(2)} =\frac{q^{1/2}a_2 {a_4}^2}{1-a_4} \frac{\Phi_{l_1}^{1,0,1} \Phi_{l_1+1}^{1,1,1}}{\Phi_{l_1}^{0,0,0} \Phi_{l_1+1}^{1,1,2}}
 \left(\frac{\Phi_{l_1}^{0,0,1}}{\Phi_{l_1}^{1,0,1}}+\frac{a_2 \big(q^{l_1+1}-a_0\big)}{q^{l_1+1/2}a_0 a_4} \frac{\Phi_{l_1+1}^{0,0,1}}{\Phi_{l_1+1}^{1,0,1}} \right),\\
 Y_{l_1}^{(3)} =-\frac{q^{1/2}a_0}{a_4} \frac{\Phi_{l_1}^{1,1,1} \Phi_{l_1+1}^{1,1,2}}{\Phi_{l_1}^{0,0,1} \Phi_{l_1+1}^{2,1,2}}
 \left(\frac{\Phi_{l_1}^{1,0,1}}{\Phi_{l_1}^{1,1,1}}+\frac{1}{q^{1/2}a_2 a_4}\frac{\Phi_{l_1+1}^{1,0,1}}{\Phi_{l_1+1}^{1,1,1}}\right),\\
 Y_{l_1}^{(4)}=\frac{q^{2 l_1+1/2}a_4}{a_2 \big(q^{l_1+1}-a_0\big)}
 \frac{\Phi_{l_1}^{0,0,1} \Phi_{l_1+1}^{1,0,2}}{\Phi_{l_1}^{1,0,1} \Phi_{l_1+1}^{0,0,1}}
 \left(\frac{\Phi_{l_1+1}^{1,0,1}}{\Phi_{l_1+1}^{1,0,2}}-\frac{\Phi_{l_1+1}^{1,1,1}}{\Phi_{l_1+1}^{1,1,2}}\right),\\
 Y_{l_1}^{(5)} =\frac{a_2 (1-a_4)}{q^{l_1}} \frac{\Phi_{l_1}^{1,0,1} \Phi_{l_1+1}^{1,1,2}}{\Phi_{l_1}^{1,1,1} \Phi_{l_1+1}^{1,0,1}}
 \left(\frac{\Phi_{l_1}^{0,0,0}}{\Phi_{l_1}^{0,0,1}} +\frac{q^{1/2}a_2 a_4}{1-a_4}\,\frac{\Phi_{l_1+1}^{1,1,1}}{\Phi_{l_1+1}^{1,1,2}}\right),
\end{gather*}
where the functions $\big\{\Phi_{l_1}^{l_0,l_2,l_3}\big\}_{l_1\in\bbZ_{\geq0}}$ are def\/ined by~\eqref{eqn:def_Phi}.
Note that the actions of translations $T_i$, $i=0,\dots,4$, on these solutions are given by the following
\begin{gather*}
 T_0\colon \ \big(a_0,a_2,a_4,l_1,q,\Phi_{l_1}^{l_0,l_2,l_3}\big) \mapsto\big(qa_0,a_2,a_4,l_1,q,\Phi_{l_1}^{l_0+1,l_2,l_3}\big),\\
 T_1\colon \ \big(a_0,a_2,a_4,l_1,q,\Phi_{l_1}^{l_0,l_2,l_3}\big) \mapsto\big(a_0,q^{-1}a_2,a_4,l_1+1,q,\Phi_{l_1+1}^{l_0,l_2,l_3}\big),\\
 T_2\colon \ \big(a_0,a_2,a_4,l_1,q,\Phi_{l_1}^{l_0,l_2,l_3}\big) \mapsto\big(a_0,q a_2,a_4,l_1,q,\Phi_{l_1}^{l_0,l_2+1,l_3}\big),\\
 T_3\colon \ \big(a_0,a_2,a_4,l_1,q,\Phi_{l_1}^{l_0,l_2,l_3}\big) \mapsto\big(a_0,a_2,q^{-1}a_4,l_1,q,\Phi_{l_1}^{l_0,l_2,l_3+1}\big),\\
 T_4\colon \ \big(a_0,a_2,a_4,l_1,q,\Phi_{l_1}^{l_0,l_2,l_3}\big) \mapsto\big(q^{-1}a_0,a_2,qa_4,l_1-1,q,\Phi_{l_1-1}^{l_0-1,l_2-1,l_3-1}\big).
\end{gather*}
\end{Corollary}

\section[Hypergeometric $\tau$ functions of $\widetilde{W}\big((A_1+A_1')^{(1)}\big)$-type]{Hypergeometric $\boldsymbol{\tau}$ functions of $\boldsymbol{\widetilde{W}\big((A_1+A_1')^{(1)}\big)}$-type}\label{section:HGtau_A1A1}

In this section, we construct the hypergeometric $\tau$ functions of $\widetilde{W}\big((A_1+A_1')^{(1)}\big)$-type.

\subsection[$\tau$ functions]{$\boldsymbol{\tau}$ functions}
The action of the transformation group $\widetilde{W}\big((A_1+A_1')^{(1)}\big)=\langle s_0,s_1,w_0,w_1,\pi\rangle$
on the para\-me\-ters~$a_0$, $a_1$ and $b$ are given by
\begin{alignat*}{5}
 &s_0\colon \ && (a_0,a_1,b) \mapsto \left(\frac{1}{a_0}, {a_0}^2 a_1, \frac{b}{a_0}\right),\qquad &&s_1\colon \ && (a_0,a_1,b)
 \mapsto \left(a_0 {a_1}^2, \frac{1}{a_1}, a_1 b\right),& \\
 &w_0\colon \ && (a_0,a_1,b) \mapsto \left(\frac{1}{a_0}, \frac{1}{a_1}, \frac{b}{a_0}\right),\qquad && w_1\colon \ && (a_0,a_1,b)
 \mapsto \left(\frac{1}{a_0}, \frac{1}{a_1}, \frac{b}{{a_0}^2 a_1}\right),& \\
 &\pi\colon \ && (a_0,a_1,b) \mapsto \left(\frac{1}{a_1}, \frac{1}{a_0}, \frac{b}{a_0 a_1}\right), &&&&&
\end{alignat*}
while its actions on the variables $\tau_i$, $i=-3,\dots,3$, are given by
\begin{alignat*}{3}
 &s_0\colon \ && (\tau_{-3},\tau_{-1},\tau_1) \mapsto
 \left(\frac{a_0 \tau_1 {\tau_{-2}}^2+\tau_{-1} \tau_0 \tau_{-2}+\tau_{-3} {\tau_0}^2}{a_0 \tau_{-1} \tau_1},
 \frac{a_0 {\tau_0}^2+b \tau_{-2} \tau_2}{a_0 \tau_1}, \frac{b \tau_{-2} \tau_2+{\tau_0}^2}{\tau_{-1}}\right),&\\
 &s_1\colon \ && (\tau_{-2},\tau_0) \mapsto \left(\frac{a_0 a_1 {\tau_{-1}}^2+b \tau_{-3} \tau_1}{a_0 a_1 \tau_0},
 \frac{a_0 {\tau_{-1}}^2+b \tau_{-3} \tau_1}{a_0 \tau_{-2}} \right),& \\
 &w_0\colon \ && (\tau_{-3},\tau_{-2},\tau_{-1},\tau_1) \mapsto \left(\tau_3,\tau_2,\tau_1,\tau_{-1}\right),& \\
 &w_1\colon \ && (\tau_{-3},\tau_{-2},\tau_0,\tau_1) \mapsto \left(\tau_1,\tau_0,\tau_{-2},\tau_{-3}\right), & \\
 &\pi\colon \ && (\tau_{-3},\tau_{-2},\tau_{-1},\tau_0,\tau_1) \mapsto \left(\tau_2, \tau_1, \tau_0, \tau_{-1}, \tau_{-2}\right),&
\end{alignat*}
where
\begin{gather*}
 \tau_2=\frac{a_0 \left(\tau_{-1} \tau_0+\tau_{-2} \tau_1\right)}{b \tau_{-3}},\qquad
 \tau_3=\frac{\tau_0 \tau_1+\tau_{-1} \tau_2}{b \tau_{-2}}.
\end{gather*}
For each element $w\in\widetilde{W}\big((A_1+A_1')^{(1)}\big)$ and function $F=F(a_i,b,\tau_j)$,
we use the notation $w.F$ to mean $w.F=F(w.a_i,w.b,w.\tau_j)$, that is, $w$ acts on the arguments from the left.
We note that the group of transformations $\widetilde{W}\big((A_1+A_1')^{(1)}\big)$
forms the extended af\/f\/ine Weyl group of type $(A_1+A_1)^{(1)}$~\cite{JNS2015:MR3403054}.
Namely, the transformations satisfy the fundamental relations
\begin{gather*}
 {s_0}^2={s_1}^2=(s_0s_1)^\infty=1,\qquad {w_0}^2={w_1}^2=(w_0w_1)^\infty=1,\\
 \pi^2=1,\qquad \pi s_0=s_1\pi,\qquad \pi w_0=w_1\pi,
\end{gather*}
and the action of $W\big(A_1^{(1)}\big)=\langle s_0,s_1\rangle$ and that of $W\big(A_1^{(1)}{}'\big)=\langle w_0,w_1\rangle$ commute.
We note that the relation $(ww')^\infty=1$ for transformations~$w$ and~$w'$ means that
there is no positive inte\-ger~$N$ such that $(ww')^N=1$.

To iterate each variable $\tau_i$, we need the translations $T_i$, $i=1,2,3$, def\/ined by
\begin{gather*}
 T_1=w_0 w_1,\qquad T_2=\pi s_1 w_0,\qquad T_3=\pi s_0 w_0.
\end{gather*}
Note that $T_i$, $i =1,2,3$, commute with each other and $T_1T_2T_3=1$. The actions of these on the parameters are given by
\begin{gather*}
T_1\colon \ (a_0,a_1,b)\mapsto(a_0,a_1,qb),\qquad T_2\colon \ (a_0,a_1,b)\mapsto\big(qa_0,q^{-1}a_1,b\big),\\
T_3\colon \ (a_0,a_1,b)\mapsto\big(q^{-1}a_0,qa_1,q^{-1}b\big),
\end{gather*}
where the parameter $q=a_0a_1$ is invariant under the action of translations. We def\/ine $\tau$ functions by
\begin{gather*}
 \tau^{l_1}_{l_2}={T_1}^{l_1}{T_2}^{l_2}(\tau_{-3}),
\end{gather*}
where $l_1,l_2\in\mathbb{Z}$. We note that
\begin{gather*}
 \tau_{-3}=\tau^0_0,\qquad \tau_{-2}=\tau^1_1,\qquad \tau_{-1}=\tau^1_0,\qquad \tau_0=\tau^2_1,\qquad \tau_1=\tau^2_0,\qquad \tau_2=\tau^3_1,\qquad \tau_3=\tau^3_0.
\end{gather*}

\subsection{Discrete Painlev\'e equations}
Let
\begin{gather*}
 f_0=\frac{\tau _{-2} \tau _1}{\tau _{-1} \tau _0},\qquad f_1=\frac{\tau _{-3} \tau _0}{\tau _{-2} \tau _{-1}},\qquad
 f_2=\frac{(\tau_{-1})^2}{\tau_{-3} \tau_1},
\end{gather*}
where
\begin{gather*}
 f_0f_1f_2=1.
\end{gather*}
The action of $\widetilde{W}\big((A_1+A_1')^{(1)}\big)$ on the variables $f_i$, $i=0,1,2$, is given by
\begin{alignat*}{3}
 &s_0 \colon \ && (f_0,f_1,f_2)\mapsto
 \left(\frac{f_0 (a_0 f_0+a_0+f_1)}{f_0+f_1+1},\frac{f_1 (a_0 f_0+f_1+1)}{a_0 (f_0+f_1+1)},\right.& \\
 &&& \left. \hphantom{(f_0,f_1,f_2)\mapsto}{}\quad \frac{a_0 f_2(f_0+f_1+1)^2}{(a_0 f_0+a_0+f_1)(a_0 f_0+f_1+1)}\right),& \\
 &s_1\colon \ && (f_0,f_1)\mapsto
 \left(\frac{f_0 (a_0 a_1+b f_0 f_1)}{a_1 (a_0+b f_0 f_1)},\frac{a_1 f_1 (a_0+b f_0 f_1)}{a_0 a_1+b f_0 f_1}\right),& \\
 &w_0\colon \ && (f_0,f_1,f_2)\mapsto \left(\frac{a_0 (f_0+1)}{b f_0 f_1},\frac{a_0 f_0+a_0+b f_0 f_1}{a_0 b f_0 (f_0+1)},
 \frac{b^2 f_0}{f_2 (a_0 f_0+a_0+b f_0 f_1)}\right),& \\
 &w_1\colon \ && (f_0,f_1)\mapsto\left(f_1,f_0\right),&\\
 &\pi\colon \ && (f_1,f_2)\mapsto\left(\frac{a_0 (f_0+1)}{b f_0 f_1},\frac{b f_1}{a_0(f_0+1)}\right).&
\end{alignat*}
By letting
\begin{gather*}
 f_{l_2}^{(0)}={T_2}^{l_2}(f_0),\qquad f_{l_2}^{(1)}={T_2}^{l_2}(f_1),\qquad f_{l_2}^{(2)}={T_2}^{l_2}(f_2),
\end{gather*}
the actions of $T_i$, $i=1,2,3$:
\begin{alignat*}{3}
 &T_1(f_1)f_1=\frac{a_0(f_0+1)}{b f_0},\qquad && T_1(f_0)f_0=\frac{T_1(f_1)+1}{b T_1(f_1)},& \\
 &T_2(f_2)f_2=\frac{b}{q f_1(f_1+1)},\qquad && T_2(f_1)f_1=\frac{a_0 (b+q T_2(f_2))}{T_2(f_2)(q a_0 T_2(f_2)+b)},& \\
 &T_3(f_0)f_0=\frac{a_1 b+q f_2}{f_2(b+q f_2)},\qquad && T_3(f_2)f_2=\frac{a_1 b}{q T_3(f_0)(T_3(f_0)+1)},&
\end{alignat*}
lead the following $q$-Painlev\'e equations
\begin{gather}
T_1\big(f_{l_2}^{(1)}\big)f_{l_2}^{(1)}=\frac{q^{l_2}a_0\big(f_{l_2}^{(0)}+1\big)}{b f_{l_2}^{(0)}},\qquad
T_1\big(f_{l_2}^{(0)}\big)f_{l_2}^{(0)}=\frac{T_1\big(f_{l_2}^{(1)}\big)+1}{b T_1\big(f_{l_2}^{(1)}\big)}, \label{eqn:A1A1_qp2_1}\\
 T_2\big(f_{l_2}^{(2)}\big)f_{l_2}^{(2)}=\frac{b}{q f_{l_2}^{(1)}\big(f_{l_2}^{(1)}+1\big)},\qquad
T_2\big(f_{l_2}^{(1)}\big)f_{l_2}^{(1)}=\frac{q^{l_2}a_0 \big(b+q T_2\big(f_{l_2}^{(2)}\big)\big)}{T_2\big(f_{l_2}^{(2)}\big)\big(q^{l_2+1} a_0 T_2\big(f_{l_2}^{(2)}\big)+b\big)},
\label{eqn:A1A1_qP_T2}\\
T_3\big(f_{l_2}^{(0)}\big)f_{l_2}^{(0)}=\frac{a_1 b+q^{l_2+1} f_{l_2}^{(2)}}{q^{l_2}f_{l_2}^{(2)}\big(b+q f_{l_2}^{(2)}\big)},\qquad
 T_3\big(f_{l_2}^{(2)}\big)f_{l_2}^{(2)}=\frac{a_1 b}{q^{l_2+1} T_3\big(f_{l_2}^{(0)}\big)\big(T_3\big(f_{l_2}^{(0)}\big)+1\big)}.
\label{eqn:A1A1_qP_T3}
\end{gather}
We note that equation~\eqref{eqn:A1A1_qp2_1} is known as a $q$-discrete analogue of Painlev\'e~II equation~\cite{KTGR2000:MR1789477}
and can be rewritten as the following single second-order ordinary dif\/ference equation~\cite{RG1996:MR1399286,RGTT2001:MR1838017,SakaiH2001:MR1882403}:
\begin{gather}\label{eqn:A1A1_qp2_2}
 \left(T_1\big(f_{l_2}^{(0)}\big)f_{l_2}^{(0)}-\frac{1}{b}\right)
 \left({T_1}^{-1}\big(f_{l_2}^{(0)}\big)f_{l_2}^{(0)}-\frac{q}{b}\right)
 =\frac{a_1}{q^{l_2}b} \frac{f_{l_2}^{(0)}}{1+f_{l_2}^{(0)}}.
\end{gather}

\subsection[Hypergeometric $\tau$ functions]{Hypergeometric $\boldsymbol{\tau}$ functions}
\label{subsection:HGtau_A1A1}
We here def\/ine hypergeometric $\tau$ functions of $\widetilde{W}\big((A_1+A_1')^{(1)}\big)$-type by $\tau^{l_1}_{l_2}$ satisfying the following conditions:
\begin{enumerate}\itemsep=0pt
\item[(i)] $\tau^{l_1}_{l_2}$ satisfy the action of the translation subgroup of $\widetilde{W}\big((A_1+A_1')^{(1)}\big)$, $\langle T_1,T_2,T_3\rangle$;
\item[(ii)] $\tau^{l_1}_{l_2}$ are functions in $b$ consistent with the action of $T_1$, i.e.,
$ \tau^{l_1}_{l_2}=\tau_{l_2}(q^{l_1}b)$;

\item[(iii)] $\tau^{l_1}_{l_2}$ satisfy the following boundary conditions:
$\tau^{l_1}_{l_2}=0$, for $l_2<0$;
\end{enumerate}
under the conditions of parameters
\begin{gather}\label{eqn:A1A1_condition_para}
 a_0=1,\qquad a_1=q.
\end{gather}

In a similar manner as Section~\ref{subsection:HGtau_A4}, we obtain the following theorem.

\begin{Theorem}\label{theorem:A1A1_hypertau}
The hypergeometric $\tau$ functions of $\widetilde{W}\big((A_1+A_1')^{(1)}\big)$-type are given by the following
\begin{gather*}
\tau^{l_1}_0=\Gamma\big(q^{l_1}b;q,q\big),\qquad
\tau^{l_1}_{l_2}=\frac{\Gamma\big(q^{l_1}b;q,q\big)}{\Theta\big(q^{l_1}b;q\big)^{l_2}} \psi^{l_1}_{l_2},
\end{gather*}
where $l_1\in\bbZ$, $l_2\in\bbZ_{>0}$ and
the functions $\big\{\psi^{l_1}_{l_2}\big\}_{l_1\in\bbZ,\, l_2\in\bbZ_{>0}}$ are given by the following $l_2\times l_2$ determinants
\begin{gather*}
 \psi^{l_1}_{l_2} =\begin{vmatrix}
 F_{l_1}& F_{l_1+1}&\cdots&F_{l_1+l_2-1}\\
 F_{l_1-1}& F_{l_1}&\cdots&F_{l_1+l_2-2}\\
 \vdots&\vdots &\ddots&\vdots\\
 F_{l_1-l_2+1}&F_{l_1-l_2+2}&\cdots&F_{l_1}
 \end{vmatrix}.
\end{gather*}
Here, the function $F_n$ is given by
\begin{gather*}
 F_n =\frac{\Theta\big({-}q^{(2n-3)/4}b^{1/2};q^{1/2}\big)}{q^n b}
 \left( A_n \,{}_1\varphi_1\left(\begin{matrix}0\\ -q^{1/2}\end{matrix};q^{1/2},-q^{(2n-5)/4}b^{1/2}\right)\right.\\
\left.\hphantom{F_n =}{} +B_n e^{\pi {\rm i}\log{b}/\log{q}}\, {}_1\varphi_1\left(\begin{matrix}0\\ -q^{1/2}\end{matrix};q^{1/2},q^{(2n-5)/4}b^{1/2}\right)
 \right),
\end{gather*}
where $A_n$ and $B_n$ are periodic functions of period one with respect to $n$, that is,
\begin{gather*}
 A_{n+1}=A_n,\qquad B_{n+1}=B_n.
\end{gather*}
\end{Theorem}

Moreover, Theorem \ref{theorem:A1A1_hypertau} leads the following corollary.

\begin{Corollary}\label{corollary:A1A1_HGsol}
Under the condition \eqref{eqn:A1A1_condition_para},
the hypergeometric solutions of $q$-Painlev\'e equations \eqref{eqn:A1A1_qp2_1}--\eqref{eqn:A1A1_qp2_2}
are given by
\begin{gather*}
 f_{l_2}^{(0)}=-\frac{1}{qb} \frac{\psi^1_{l_2+1}\psi^2_{l_2}}{\psi^1_{l_2}\psi^2_{l_2+1}},\qquad
 f_{l_2}^{(1)}=q^{l_2+1} \frac{\psi^0_{l_2}\psi^2_{l_2+1}}{\psi^1_{l_2+1}\psi^1_{l_2}},\qquad
 f_{l_2}^{(2)}=-\frac{b}{q^{l_2}} \frac{(\psi^1_{l_2})^2}{\psi^0_{l_2}\psi^2_{l_2}}.
\end{gather*}
Note that the actions of translations $T_i$, $i=1,2,3$,
on these solutions are given by the following
\begin{gather*}
 T_1\colon \ \big(b,q,\psi^{l_1}_{l_2}\big) \mapsto\big(qb,q,\psi^{l_1+1}_{l_2}\big),\qquad
T_2\colon \ \big(b,q,\psi^{l_1}_{l_2}\big) \mapsto\big(b,q,\psi^{l_1}_{l_2+1}\big),\\
T_3\colon \ \big(b,q,\psi^{l_1}_{l_2}\big) \mapsto\big(q^{-1}b,q,\psi^{l_1-1}_{l_2-1}\big).
\end{gather*}
\end{Corollary}

\subsection*{Acknowledgements}

The author would like to express his sincere thanks to Dr.~Milena Radnovic for her valuable comments.
This research was supported by grant \# DP130100967 from the Australian Research Council.

\pdfbookmark[1]{References}{ref}
\LastPageEnding

\end{document}